\documentclass[lettersize,journal]{IEEEtran}
\usepackage{amsmath,amsfonts}
\usepackage{algorithmic}
\usepackage{algorithm}
\usepackage{array}
\usepackage{textcomp}
\usepackage{url}
\usepackage{verbatim}
\usepackage{graphicx}
\usepackage{cite}
\usepackage{hyperref}
\usepackage{colortbl}

\usepackage{mathrsfs}

\usepackage{multicol}
\usepackage{multirow}
\usepackage{cases}

\usepackage{subfigure}
\newtheorem{theorem}{{Theorem}}
\newtheorem{lemma}{{Lemma}}

\newtheorem{remark}{{Remark}}
\newtheorem{corollary}{{Corollary}}

\begin{document}
 
\title{Enhancing AAV-Enabled Secure Communications via Synthetic Aperture Beamforming}

\author{Bin~Qiu,~\IEEEmembership{Member,~IEEE,}
	Wenchi~Cheng,~\IEEEmembership{Senior Member,~IEEE,}
	Hongxiang~He,\\
	and~Jiangzhou~Wang,~\IEEEmembership{Fellow,~IEEE}
     


%
	
\thanks{Bin Qiu, Wenchi Cheng, and Hongxiang He are with the State Key Laboratory of Integrated Services Networks, Xidian University,
		Xian 710071, China (e-mail: qiubin@xidian.edu.cn; wccheng@xidian.edu.cn; hehongxiang@stu.xidian.edu.cn).}

\thanks{Jiangzhou Wang is with the School of Information Science and Engineering, Southeast University, and Purple Mountain Laboratories, Nanjing 211119, China (e-mail: j.z.wang@kent.ac.uk).}
}

\maketitle

\begin{abstract}
In this paper, we consider a synthetic aperture secure beamforming approach for a virtual multiple-input multiple-output (MIMO) broadcast channel in the presence of hybrid wiretapping environments. Our goal is to design the flight node deployment constructed by a single-antenna mobile autonomous aerial vehicle (AAV), corresponding transmission symbol strategy, transmit precoding, and received beamforming to maximize the system channel capacity. Leveraging the synthetic aperture beamforming, we aim to provide spatial gain along a predefined angle in free space while reducing it in others and thus enhance physical layer (PHY) security. To this end, we analyze the expression of the asymptotic channel eigenvalues to optimize the AAV flight node deployment. For the optimal precoding design, an energy-efficient method that minimizes the transmit power consumption is studied based on the given virtual MIMO channel, while meeting the quality of service (QoS) for the base station (BS), leakage tolerance of eavesdroppers (Eves), and per-node power constraints. The power minimization problem is a non-convex program, which is then reformulated as a tractable form after some mathematical manipulations. Moreover, we design the received beamforming by applying the linearly constrained minimum variance (LCMV) method such that the jamming can be effectively suppressed. Numerical results demonstrate the superiority of the proposed method in promoting capacity.
\end{abstract}


\section{Introduction}\label{sec:Intro}
\IEEEPARstart{N}{owadays}, various aircraft, such as satellites, airships and autonomous aerial vehicles (AAVs), have significantly expanded the scope of application fields, like power inspection, emergency response, and smart city~\cite{9598918}. AAVs, in particular, have garnered considerable research interest due to the strong maneuverability and prominent deployment flexibility. By integrating AAVs into wireless networks, the air-to-ground (A2G) links exhibit a high probability of line-of-sight (LoS) components, which indicates a promising enhanced system performance over conventional terrestrial cellular networks~\cite{9456851}. Along with the countless applications of AAVs, their links are vulnerable to security threats from malicious attackers due to the inherent broadcast nature of the communication medium. As a result, security has emerged as a pivotal concern in the implementation and operation of both current and future AAV communication systems~\cite{8883127,11244297}.

Against this background, physical layer (PHY) security is a supplement to traditional encryption technology based on information theory and can improve security by utilizing the time-varying characteristics of the wireless medium to degrade or eliminate the wireless signals received by adversaries. In contrast to key-based cryptography applied to higher layers, PHY security protects information transmission without relying on secret keys or complex algorithms. The foundation of PHY security was laid by Wyner's pioneering work~\cite{6772207}, which demonstrated the possibility of secrecy in a single-input single-output (SISO) wiretap channel. Subsequently, this concept was expanded to encompass the Gaussian wiretap channel~\cite{1055892}. 

Up to now, there are several works that have concentrated on enhancing the security performance of AAV-enabled communication networks. In~\cite{10058144}, a robust and secure task transmission and computation scheme was considered for multi-antenna AAV-assisted mobile edge computing networks, where an AAV serves as both mobile edge computing and relay functions. By integrating sensing into AAV networks, the authors in~\cite{10345500} proposed an integrated sensing, navigation, and communication framework to guarantee PHY security in the presence of a mobile eavesdropper (Eve).  An information AAV serving as an aerial base station transmits confidential messages to multiple legitimate ground users and jams an Eve with artificial noise (AN) to facilitate secure communication. Leveraging the real-time and robust characteristics of model predictive control effectively addresses the imprecision in AAV positioning caused by disturbances. In view of this, a problem of navigating an AAV from a specified starting point to a designated endpoint was explored to enhance the system robustness, while considering factors such as communication security for users, power consumption, and AN~\cite{11069265}.

Due to the limitations of single-AAV size, weight, and power, it may fail to provide satisfactory performance, which is specifically required for AAV links like remote communication and high-level security. As a result, it is of interest to develop a multi-AAV cooperation mode that provides extra freedom to achieve higher communication capability~\cite{10439615}. In~\cite{9045989}, the authors investigated a resource allocation policy to enhance fairness in secure communication facilitated by multiple AAVs-enabled communications, where multipurpose AAVs are dispatched to provide PHY security for the legitimate receivers via the assigned subcarriers while idle AAVs are served as jammers for satisfactory secrecy performance provisioning. To fully exploit the benefit of multi-AAV flexibility, a novel AAV swarm-enabled collaboration mechanism was investigated in~\cite{8469055}, where a virtual multiple-antenna entity formed by a group of single-antenna AAVs yields improvement in diversity and communication reliability. Inspired by this work, the authors in~\cite{10536043,10382648,10812989,10759093} explored an AAV swarm-enabled virtual antenna array framework to transmit data via collaborative beamforming mechanism capable of enhancing the signal gain and directivity. Moreover, through a virtual multi-antenna system formed by multiple AAVs, the authors in~\cite{9968197}  proposed a multi-AAV cooperative sensing and transmission scheme with overlapped sensing task allocation.

Most of the existing PHY security studies focused on passive eavesdropping, where the anti-eavesdropping technique evolution was the main research goals. New challenges are posed when intelligent hybrid attackers exploit full-duplex technology~\cite{8610381,10417007,7470273}. Such attackers, also referred to as active Eves, are prone to jeopardizing the privacy of communication via hybrid wiretapping scenarios that perform both eavesdropping information and malicious jamming legitimate channels. The jamming tries to degrade the quality of the signals received by legitimate receivers, forcing the transmitter to increase its power to ensure reliable transmission, which in turn is more conducive to wiretapping, and thus the potential of this advanced technology can be conducted to launch more hazardous attacks~\cite{8674764}. Consequently, this poses new security challenges for AAV-enabled communications, but works on hybrid wiretapping in AAV-enabled communication systems are relatively rare at present.

Despite the practical advantages of AAV-enabled communications, several technical challenges should be addressed to unlock the promising potential for performance improvements. Typically, the energy capacity of AAVs relying on onboard batteries is limited. Hence, one key bottleneck is that the strict power constraints hinder the promotion of high-performance AAV-enabled communications. Energy-efficient AAV-enable secure communications attract significant research interest. The integration of AN consumes part of the transmit power thereby reducing the power efficiency, and thus the AN-aided method is not suitable for the application in AAV scenarios. Given the AAV’s size and weight constraints in practice, it is impossible for an AAV to carry large-scale antenna arrays, thus failing to provide satisfactory spatial degrees of freedom (DoFs)~\cite{9848802}. For multi-AAV cooperation scheme, it is generally required explicit information sharing among multiple AAVs before collaborative transmission, which leads to extra time and power consumption. Moreover, it is very hard to guarantee strict time, frequency, and trajectory synchronization control due to the AAVs' individual heterogeneity.

Motivated by the above observations, in this work, we consider an AAV-enabled secure transmission approach, where a source transmits to a multi-antenna base station (BS) in the presence of a hybrid wiretapping environment. Different from the previous collaborative beamforming works, we investigate the scenario where a virtual non-uniform linear array (NULA) arranged by a single-antenna mobile AAV forms synthetic aperture beamforming, and combines it with corresponding transmission symbol strategy to provide spatial division gain. It is important to note that our proposed synthetic aperture beamforming approach not only overcomes the size, weight, and power constraints of a single AAV, but also avoids the application challenges of  multi-AAV complex collaboration, such as interconnection, multi-object control, and individual heterogeneity. To the best of the authors’ knowledge, such a technique to achieve secure transmission has not been investigated in the literature, yet. For the present problem, our goal is to maximize the system channel capacity by the judicious design of the transmit symbol strategy, AAV flight node positions (virtual array deployment), transmit precoding, and received beamforming to guarantee reliable and secure AAV communications. To facilitate the maximization of the channel capacity, we first explore the formulated problem structure to get some insights into the optimal solution. By leveraging the insights, we found that the original problem can be transformed into two independent problems. For the optimal AAV node deployment, we derive analytical expressions for the eigenvalues of virtual multiple-input multiple-output (MIMO) channels. Through this derivation, we demonstrate that the asymptotically optimal AAV node deployment is associated with the Fekete-point distribution. Afterward, an energy-efficient precoder is devised using the given channel matrix to minimize the transmit power while providing quality of service (QoS) guarantee, satisfying leakage tolerance, and meeting per-node power constraints. For tractability, we reformulate the problem by converting the non-convex constraints into linear matrix inequality (LMI) terms, and then the resulting problem is cast as a semi-definite program (SDP). Next, the received beamforming is resorted to the linearly constrained minimum variance (LCMV) method to eliminate jamming. Last, we extend our proposed scheme to more practical cases, such as high-dimensional virtual array and robust adaptive beamforming.

The rest of this paper is organized as follows. Section~\ref{sec:Model} describes the system model and problem formulation. In Section~\ref{sec:Opti}, some insights into the considered problem are given; then, the optimization problem is solved to obtain the optimal AAV node deployment, transmit precoding, and received beamforming, followed by some extensions of more realistic scenarios in Section~\ref{sec:Exte}.  In Section~\ref{sec:Simu}, numerical results are provided to show the superiority of our proposed scheme, and Section~\ref{sec:Conc} concludes the paper. Some details regarding the derivations are provided in the Appendices.

\textit{Notations}: Throughout this paper, bold capital and lower-case letters denote the matrices and vectors, respectively. We represent the expectation and the trace operators as $\mathbb{E}\{\cdot\}$ and Tr$(\cdot)$. Here, the superscripts $( \cdot )^ {-1} $, $( \cdot )^T$, and $( \cdot )^H$ represent inverse, transpose, and Hermitian operations, respectively. The denotations of Euclidean norm and absolute value are given by $\left\|  \cdot  \right\|_2$ and $\left| \cdot \right|$, respectively. We use ${[ \cdot ]_{i}}$ and ${[ \cdot ]_{i,j}}$ to indicate the $i$th element of the vector and the entry in the $i$th row and the $j$th column of the matrix, respectively. ${\rm{diag}}(\cdot)$ denotes a diagonal matrix with the elements of a vector on the main diagonal. ${\mathfrak{Re}}\{\cdot\}$ and ${\mathfrak{Im}}\{\cdot\}$ denote the real part and imaginary part of the corresponding arguments, respectively. We use ${{\mathbf{I}}_N}$ to denote an $N$-dimensional identity matrix. $\mathbb{R}$, $\mathbb{C}$, and $\mathbb{H}$ represent the real, complex, and Hermitian spaces, respectively.

\section{System Model and Problem Formulation}\label{sec:Model}
In this section, after describing the AAV-enabled communication system model, we introduce the problem statement of interest.
\subsection{Channel Model}

\begin{figure}[!t]
	\centering
	\includegraphics[width=0.75\columnwidth]{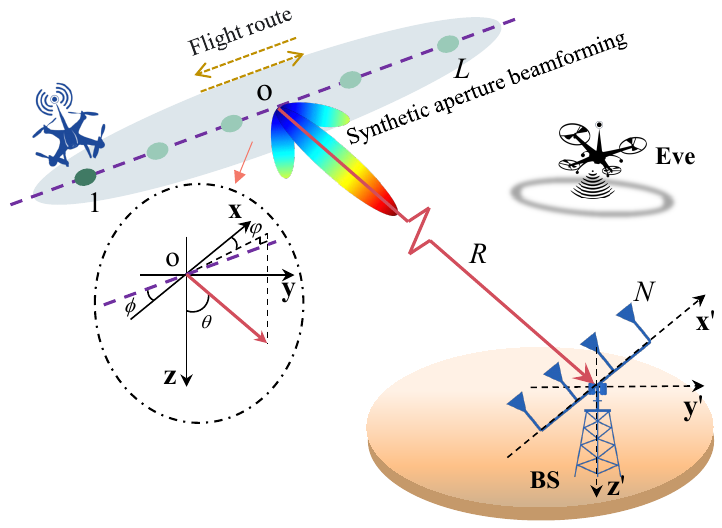}
	\caption{AAV-enabled synthetic aperture secure transmission systems.}
	\label{fig_1}
\end{figure}

Our system model, as shown in Fig.~\ref{fig_1}, considers an A2G communication system, where a rotary-wing AAV equipped with a single-antenna is deployed as a transmitter and provides wireless service to a remote multi-antenna BS in the presence of several passive Eves and $Q$ adversaries with the dual capability of eavesdropping as well as jamming any ongoing transmission to degrade the secrecy performance. A virtual array constructed by the mobile AAV along a predetermined route to form synthetic aperture beamforming, which can provide spatial gain towards the BS. More specifically, the AAV flies back and forth along a straight route at variable speeds to form $L$ flight nodes, thereby constructing a virtual NULA with $L$-element. It is assumed that the length of the one-way flight route is $D$, which represents the aperture size of the virtual transmit array. The BS employs an isotropic $N$-element uniform linear array (ULA), and it together with the virtual transmit array forms a virtual MIMO transmission system~\cite{8411465}. 

Without loss of generality, a three-dimensional (3D) Cartesian coordinate is employed to denote the precise positions of the virtual elements deployed by the AAV flight nodes. Assume that the virtual transmit NULA lies in the $x$–$y$ plane, which is centered at the origin and is parallel to the ground plane. The $x$-axis is aligned with the BS's ULA, and the $z$-axis is selected to point towards the ground. For simplicity, it is assumed that the AAV maintains a constant altitude, whose flight altitude is sufficiently high to establish a point-to-point far-field LoS link with the BS.$\footnote{Based on field measurements, for an AAV with a flight altitude of 100 meters and a cell with a radius of 600 meters, the link is guaranteed to be LoS channel~\cite{8974403}. Moreover, the AAV's flight altitude can be adjusted according to terrain type and cell size so that the LoS probability of the A2G channel approaches one.}$ The range between the centers of transmit-receive arrays is $R$. Let $\varphi $ and $\theta $ denote the corresponding azimuth and elevation angles of departure from the virtual array toward the BS, respectively. Denote by $\phi$ the rotation offset which refers to the angle between the virtual transmit NULA and the $x$-axis. Driven by the secure transmission with low probability of interception (LPI), we introduce random rotation offsets to form dynamic stochastic channels, which randomly scramble the received signals of Eves. For the convenience of describing the AAV node deployment, we use $\{\delta_l\}_{l \in {\cal L}} \in [-1,1]$ to denote the normalized spacing on the transmit virtual array relative to the center with ${\cal L} \buildrel \Delta \over =[1,2,...,L]$. Based on the basic geometric manipulations, a closed-form relationship between the angles and corresponding coordinates of the $l$th AAV node, denoted by $\left({x_{A,l}},y_{A,l},z_{A,l}\right)$, can be derived as

\begin{align}
	\left\{ \begin{array}{l}
		{x_{A,l}} = \frac{{D\delta_l \cos \phi }}{2},\\
		{y_{A,l}} = \frac{{D\delta_l \sin \phi }}{2},\\
		{z_{A,l}} = 0.
	\end{array} \right.
	\label{eq1}\end{align}
Likewise, the specific coordinates of the $n$th element at the ground BS, $\forall n\in {\cal N} \buildrel \Delta \over =[1,2,...,N]$, denoted by $\left(x_{G,n},y_{G,n},z_{G,n}\right)$, is computed as
\begin{align}
	\left\{ \begin{array}{l}
		{x_{G,n}} = \frac{{2n - 1 - N}}{2}d + R\sin \theta \cos \varphi,\\
		{y_{G,n}} = R\sin \theta \sin \varphi,\\
		{z_{G,n}} = R\cos \theta,
	\end{array} \right.
	\label{eq2}\end{align}
where $d=c/2f_c$ denotes the inter-element spacing of ULA at the ground BS to avoid aliasing effects with $f_c$ and $c$ denoting the carrier frequency and the speed of light. 

In accordance with ray tracing principles~\cite{4155681}, the channel coefficient between each pair of transmit and receive antennas  is related to the radio wave propagation range, whose value is determined by 
\begin{align}
	{h_{n,l}} =\rho ({\tau_{n,l}}) {e^{j{2\pi f_c}\left(t - {\frac{{{\tau_{n,l}}}}{c}}\right)}}, \quad \forall n\in {\cal N},\forall l\in {\cal L}, 
	\label{eq3}\end{align}
where $\rho ({\tau_{n,l}})={c}/{{4\pi f_c {\tau_{n,l}}}}$ denotes the signal attenuation factor with respect to the transmission range, and ${\tau_{n,l}}$ indicates the radio wave propagation range between the $l$th transmit AAV node to the $n$th received element, which satisfies$\footnote{By employing a first-order Maclaurin series expansion~\cite{Mathematics}, i.e. ${(1 + x)^{1/2}} \approx 1 + x/2$, the approximation is proper.}$
\begin{eqnarray}
	\begin{aligned}[b]
	{\tau _{n,l}}&\! =\! \sqrt {{{({x_{G,n}} \!-\! {x_{A,l}})}^2} \!+ \!{{({y_{G,n}} \!-\! {y_{A,l}})}^2} \!+\! {{({z_{G,n}}\! -\! {z_{A,l}})}^2}} \\&
    \approx  R  \!+ \! \frac{{{{(2n \! -  \!1 \! - \! N)}^2}{d^2}}}{{8R}} \! + \! \frac{{(2n \! - \! 1 - \! N)d\sin \theta \cos \varphi }}{2} \\&
      \kern 12pt    - \frac{{2n - 1 - N}}{{4R}}dD{\delta_l}\cos \phi - \frac{{D{\delta_l}\sin \theta \cos \varphi \cos \phi }}{2} \\&
      \kern 12pt    - \frac{{D{\delta_l}\sin \theta \sin \varphi \sin \phi }}{2} + \frac{{{{(D{\delta_l})}^2}}}{{8R}}.
	\end{aligned}
	\label{eq4}\end{eqnarray}
Since it is assumed to be far-field transmission, the difference in attenuation among elements can be ignored. To elucidate the synthetic aperture beamforming characteristics, we have made some assumptions to simplify the formulation, e.g., perfectly synchronized in time and frequency between the transmit-receive end. We assume that the AAV and BS devices are equipped with global positioning system (GPS) modules to obtain information regarding their own locations. During the transmission, the BS sends acknowledgment packets to inform the AAV of successful reception of information packets. Therefore, it is assumed perfect knowledge of the AAV-to-BS channel matrix during the whole transmission period.$\footnote{Actually, AAV is impaired by these unavoidable uncertain, such as body jitter, and wind disturbances. As an independent subject, robust synthesis schemes of synthetic aperture beamforming is a crucial issue that drives practical applications. Due to space constraints, it is out of our interest in this paper.}$  Inserting \eqref{eq4} into \eqref{eq3}, the virtual MIMO channel matrix is then converted as
\begin{eqnarray}
	\begin{aligned}[b]
	{\bf{H}}  = \rho ({R}){{\bf{B}}_G}{\bf{\tilde H}}{{\bf{B}}_A},
	\end{aligned}
\label{eq5}\end{eqnarray}
where ${{\bf{B}}_A} = {\rm{diag}}\left(\{{b_{A,l}}\}_{l \in {\cal L}}\right)\in \mathbb{C}^{L \times L}$ and ${{\bf{B}}_G} = {\rm{diag}}\left(\{{b_{G,n}}\}_{n \in {\cal N}}\right)\in \mathbb{C}^{N \times N}$ with the diagonal entries satisfying
\begin{align}
	&{b_{A,l}} = {e^{-j\frac{{2\pi f_c}}{c}\left[ {\frac{{{{(D{\delta_l})}^2}}}{{8R}} - \frac{{D{\delta_l}\sin \theta \cos \varphi \cos \phi }}{2} - \frac{{D{\delta_l}\sin \theta \sin \varphi \sin \phi }}{2}}\right]}},\label{eq6}\\&
	{b_{G,n}} = {e^{-j\frac{{2\pi f_c}}{c}\left[\frac{{{{(2n - 1 - N)}^2}{d^2}}}{{8R}} + \frac{{(2n - 1 - N)d\sin \theta \cos \varphi }}{2}\right]}},\label{eq7}
\end{align}
and ${\bf{\tilde H}} = \{ {\tilde h_{n,l}}\}_{n \in {\cal N}, l \in {\cal L}}\in \mathbb{C}^{N \times L}$ whose entities are given by
\begin{eqnarray}
	\begin{aligned}[b]
		{\tilde h_{n,l}} = {e^{j\omega  \frac{{(2n - 1 - N)}}{N}{\delta_l}\cos \phi }},\quad \forall n\in {\cal N},\forall l\in {\cal L},
	\end{aligned}
	\label{eq8}\end{eqnarray}
with $\omega  = \frac{\pi f_c}{c} \cdot \frac{DNd}{2R}$. Notice that we have assumed far-field model for AAV-enabled communications, and thus the communication distance is sufficiently large compared to the product of the antenna aperture sizes of the transmit and receive arrays, satisfying $R \gg DNd$, which will provide us with some insights into the asymptotic channel behavior as $\omega$ is closer to zero. This facilitates later analysis of the considered problem features.

\subsection{Symbol Transmission Strategy}
\begin{figure}[!t]
	\centering
	\includegraphics[width=0.75\columnwidth]{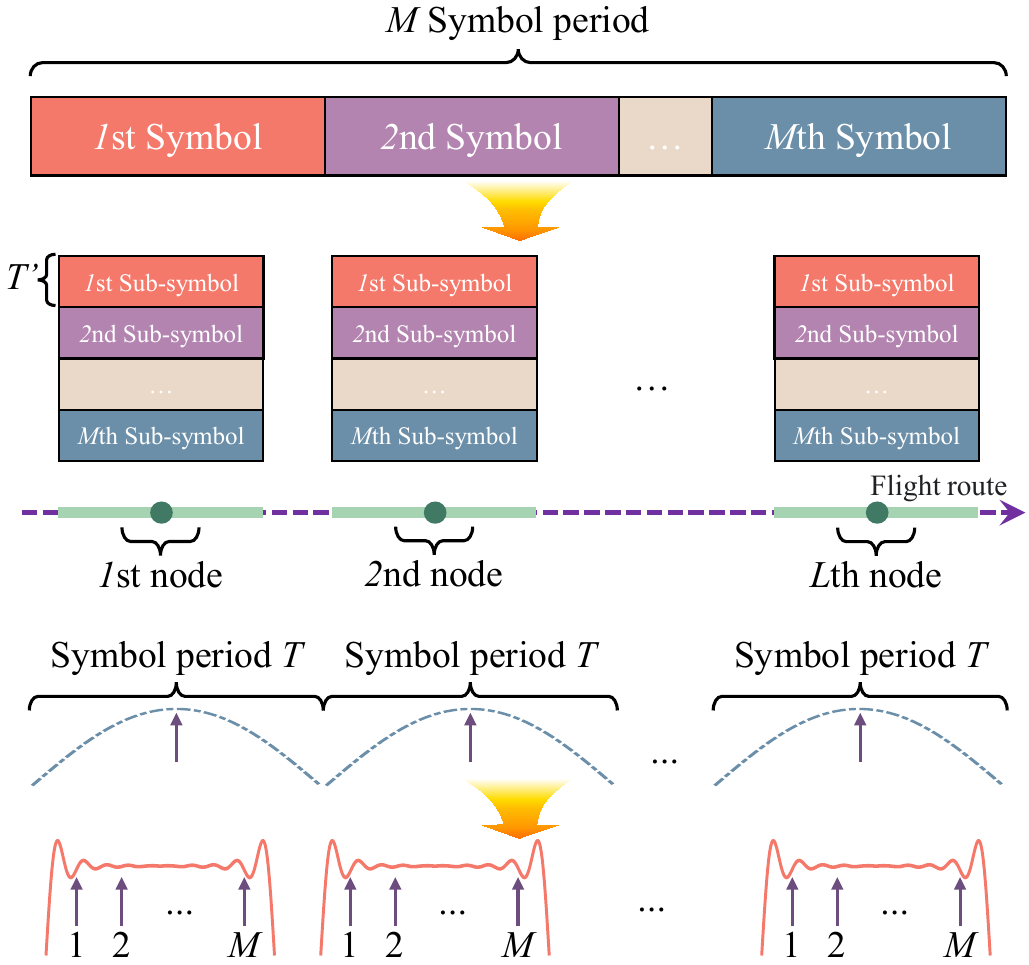}
	\caption{Illustration of the synthetic aperture transmission symbol strategy.}
	\label{fig_2}
\end{figure}

To provide spatial diversity gain, a corresponding symbol transmission strategy should be studied to match the virtual array constructed by a mobile AAV. In particular, it is supposed that an AAV tries to transmit a serial sequence to the ground BS, as illustrated in Fig.~\ref{fig_2}, with each symbol period being $T$. The flight time between two nodes is set as a fixed interval $\Delta_T $. Divide the $M$ serial transmit symbols into one group as a sub-symbol block, denoted by $\{{\hat x_m}\}_{m \in {\cal M}}$, with each sub-symbol period being $T'<T$, i.e., ${x_m} \to {\hat x_m}, \forall m \in {\cal M} \buildrel \Delta \over =[1,2,...,M]$, and repeat the transmission of this one sub-symbol block at each node so that the same data stream can be combined at the receiver to get spatial diversity gain. To implement more effectively, let us assume that the AAV flies continuously along its flight route, and a certain position around one node is regarded as its node region, which is a reasonable assumption since the AAV speed is much less than the sub-symbol rate, i.e., $T' \ll \Delta_T$.$\footnote{Theoretically, we can precisely control the AAV to hover at each node position until a sub-symbol block is launched. In practical applications, there exists a trade-off between performance and implementation cost.}$ During the $l$th node region, the corresponding transmit signal can be expressed as
\begin{eqnarray}
	\begin{aligned}[b]
	{s_{m,l}} = {u_l}{\hat x_m},\quad \forall m\in {\cal M},\forall l\in {\cal L},
	\end{aligned}
\label{eq9}\end{eqnarray}
where ${u_l}$ is the transmit precoder when AAV falls in the $l$th node region, and ${\hat x_m}$ is the $m$th transmit sub-symbol. Let us stack all transmit precoders as a precoding vector ${\bf{u}}$. After flying through all $L$ AAV nodes, the $m$th transmit signal vector is given by
\begin{eqnarray}
	\begin{aligned}[b]
	{{\bf{s}}_m} = {\bf{u}}{\hat x_m},\quad \forall m\in {\cal M}.
	\end{aligned}
\label{eq10}\end{eqnarray}
\begin{remark} \label{remark01}
Through the AAV node deployment optimization results in the next section, we know that the AAV flight nodes construct a virtual transmit NULA. Assuming a fixed flight time interval between two nodes, the sub-symbol block is repeatedly transmitted at the corresponding uniform intervals. Alternatively, the AAV flies at a constant speed and repeatedly transmits sub-symbol blocks upon arrival at the designated node positions.
\end{remark}

\subsection{Received Signal Model}
Define $\left({x_{E,q}},y_{E,q},z_{E,q}\right)$ as the coordinates of the $q$th Eve, $\forall q\in {\cal Q}$. By the basic geometric manipulations, the corresponding angle of arrival ${\vartheta _{J,q}}$ from the $q$th Eve to the ground BS can be computed, and we further obtain the jamming steering vector as
\begin{eqnarray}
	\begin{aligned}[b]
	 {\bf{a}}({\vartheta _{J,q}}) = {\left[ {{e^{j\frac{{2\pi {f_c}}}{c}{{{\zeta }}_{J,q}}(1)}},...,{e^{j\frac{{2\pi {f_c}}}{c}{{{\zeta }}_{J,q}}(N)}}} \right]^H},
	\end{aligned}
\label{eq11}\end{eqnarray}
where ${{{\zeta }}_{J,q}}(n) = [n - (N + 1)]d\sin {\vartheta _{J,q}}/2$, $\forall n\in {\cal N}$, denotes the phase term. For the sake of notation simplification, ${\bf{a}}_{J,q}$ is used to denote the associated channel vector, i.e., ${\bf{a}}_{J,q}\buildrel \Delta \over = {\bf{a}}({\vartheta _{J,q}})$. Then, the jamming signal from the $q$th Eve can be expressed as 
\begin{eqnarray}
	\begin{aligned}[b]
		{{\bf{s}}_{J,q}} = \sqrt {{P_q}} {{\bf{a}}_{J,q}}{x_J},\quad \forall q\in {\cal Q},
	\end{aligned}
	\label{eq12}\end{eqnarray}
where ${x_J}$ denotes the random signals with $\mathbb{E}\{|{x_J}|^2\}=1$, and $\sqrt {{P_q}}$ is the $q$th jamming power arriving at the BS.

Hereby, the virtual MIMO transmission for the $m$th transmit signal vector at the BS is modeled as
\begin{eqnarray}
	\begin{aligned}[b]
	{{\bf{y}}_{G,m}} = {\bf{H}}{{\bf{s}}_m} +  {\sum\limits_{q \in {\cal Q}} {{\bf{s}}_{J,q}}} + {{\bf{n}}_G},\quad \forall m\in {\cal M},
	\end{aligned}
\label{eq13}\end{eqnarray}
where ${{\bf{n}}_G}\in \mathbb{C}^{N \times 1}$ represents the complex additive white Gaussian noise (AWGN) satisfying ${\bf{n}}_G\sim{\mathcal{CN}}({{\boldsymbol{0}}},\sigma_G^2{\bf{I}}_N)$. 
\begin{remark} \label{remark02}
This key observation in equation \eqref{eq13} is that the formulation structurally resembles a multi-antenna transmission system, which implies that our design obtains the multi-antenna gain through a single-antenna mobile AAV’s sequential operation across different sub-symbol periods. 
\end{remark}

 Performing received beamforming processing, the BS obtains the signals transmitted from $l$th AAV node as ${{\tilde y}_{G,m,l}} = {{\bf{w}}^H}{\left[ {\bf{H}} \right]_{:,l}}{s_{m,l}},\forall l\in {\cal L}$. After storing all $L$ AAV node transmit signals, we combine the data stream as
\begin{eqnarray}
	\begin{aligned}[b]
	{{\tilde y}_{G,m}}& = {{\bf{w}}^H}{{\bf{y}}_{G,m}}\\&
	= \underbrace {{{\bf{w}}^H}{\bf{Hu}}{\hat x_m}}_{{\text{Sub-symbol}}} + \underbrace {\sum\limits_{q \in {\cal Q}} {{{\bf{w}}^H}} {{\bf{s}}_{J,q}}}_{{\rm{Jamming}}} + \underbrace {{{\bf{w}}^H}{{\bf{n}}_G}}_{{\rm{Noise}}}, \quad \forall m\in {\cal M},
	\end{aligned}
\label{eq14}\end{eqnarray}
where ${{\bf{w}}}\in \mathbb{C}^{N \times 1}$ denotes the received beamformer complex weight coefficients (called a beamvector), which is designed for the interference suppression. It can be seen that the received signals consist of all $m$th sub-symbol, jamming signals, and thermal noise.

The ability to operate in a full-duplex mode enables the Eves to concurrently execute both attacks. Due to the emission, Eves can be discovered. According to the coordinates of the $q$th Eve, we can calculate the azimuth angle of departure from the virtual transmit array to the $q$th Eve, denoted by ${\vartheta _{E,q}}$. However, the locations of the Eves are uncertain for the uncontrollable factors, e.g., measurement error and irregular movement. The design should address eavesdropping vulnerability under imperfect Eve channel information. Hence, a robust secure beamforming strategy is proposed by leveraging a moment-based random model where estimates of Eve's first and second order statistics are acquirable~\cite{6781609}. Likewise, we define the actual channel vector from virtual transmit array to the $q$th Eve as ${\bf{h}}_{E,q}\buildrel \Delta \over = {\bf{h}}({\vartheta _{E,q}},R_{E,q})$, which is given by
\begin{align}
	{{\bf{h}}_{E,q}} =\rho (R_{E,q}){\left[{e^{ - j\frac{{2\pi {f_c}}}{c}{{{\zeta }}_{E,q}}(1)}},...,{e^{ - j\frac{{2\pi {f_c}}}{c}{{{\zeta }}_{E,q}}(L)}}\right]^H},
	\label{eq15}\end{align}
with ${{{\zeta }}_{E,q}}(l) = {\delta _l}D\sin {\vartheta _{E,q}}$, $\forall l\in {\cal L}$, being the phase term. The deterministic model with uncertainty is given by
\begin{eqnarray}
	\begin{aligned}[b]
	{{\bf{h}}_{E,q}} = {{\bf{\tilde h}}_{E,q}} + \Delta {{\bf{h}}_{E,q}}, \quad \forall q\in {\cal Q},
	\end{aligned}
\label{eq16}\end{eqnarray}
where ${{\bf{\tilde h}}_{E,q}}$ denotes the channel estimate of the $q$th Eve, and $\Delta {{\bf{h}}_{E,q}}$ is the uncertainty. 

The uncertainty associated with different Eves is assumed to be independent and have equal variance. Let ${{\cal U} _q}$ represent the set of all possible channel uncertainties for the $q$th Eve, i.e.,
\begin{align}
	{{\cal U} _q}  = \left\{ {\Delta {{\bf{h}}_{E,q}}{\in \mathbb{C}^{L \times 1}}: \Delta {\bf{h}}_{E,q}^H\Delta {{\bf{h}}_{E,q}} \!\le\! \epsilon _{E,q}^2} \right\},\quad \forall q\in {\cal Q},
	\label{eq17}\end{align}
where $\epsilon_{E,q}>0$ is the extent of the uncertainty radius. Assuming a worst-case eavesdropping scenario, the jamming can be null via collusive wiretapping~\cite{8610381}. Consequently, the received signal of the $q$th Eve is then denoted by
\begin{eqnarray}
	\begin{aligned}[b]
		{y_{E,q}} = {\bf{h}}_{E,q}^H{{\bf{s}}_m} + {n_{E,q}}, \quad \forall q\in {\cal Q},\forall m\in {\cal M},
	\end{aligned}
	\label{eq18}\end{eqnarray}
where ${{{n}}_{E,q}}$ denotes the complex AWGN of the $q$th Eve satisfying ${n_{E,q}}\sim{\mathcal{CN}}({{{0}}},\sigma _{E,q}^2)$.

\subsection{Problem Statement}
Assume that the virtual MIMO channel supports $K (K \le L)$ parallel data stream transmission at most. According to the information-theoretic principles of MIMO systems~\cite[Ch. 7.1.2]{tse2005}, the best way is to carry signals along the primary $K$ eigenmodes of the channel. It is widely acknowledged that Claude Shannon introduced information theory to define the boundaries of reliable communication~\cite{4460100604}. Aligned with this pursuit, our objective is to establish a reliable transfer service for the A2G link under jamming conditions, while concurrently safeguarding information against eavesdropping. The virtual array deployment refers to the geometric arrangement of AAV flight nodes, which plays an essential role in the beamforming capabilities, directivity, and gain. Viewing this fact, we integrate the AAV node deployment, transmit precoding, and received beamforming, yielding a novel synthetic aperture beamforming transmission strategy to improve the transmission performance. Explicitly, the original problem (OP) for maximizing channel capacity can be formulated as
\begin{align}
	\text{(OP)}: \quad \mathop {\max }\limits_{\{{\boldsymbol{\delta}},{\bf{u}},{\bf{w}}\} } ~ C = \sum\limits_{k \in{\cal K}} {{{\log }_2}} \left( {1 + \frac{\gamma }{K}{\lambda_k}} \right),
	\label{eq19}\end{align}	
where ${\lambda_k}$ indicates the $k$th largest eigenvalue of the channel gain matrix ${\bf{G}} \in \mathbb{H}^L$, which is defined as ${\bf{G}} \buildrel \Delta \over = {\bf{H}}^H{{\bf{H}}}$,$\footnote{Here we suppose $N \ge L$. Or else, the channel gain matrix is ${\bf{G}} \buildrel \Delta \over = {\bf{H}}{{\bf{H}}^H}$ when $L > N$.}$ ${\cal K} \buildrel \Delta \over = [1,2,...,K]$, $\boldsymbol{\delta }\buildrel \Delta \over=[\delta_1,\delta_2,...,\delta_L]^T \in \mathbb{R}^{L \times 1}$ is normalized spacing vector of the virtual transmit NULA, and ${\gamma}$ represents the prescribed received signal-to-noise ratio (SNR) requirements for providing QoS assurance. 

\section{Optimal Solution of The Problem}\label{sec:Opti}
In this section, we develop an algorithm that finds the optimal solution for the channel capacity maximization problem by utilizing the unique property of the channel.

\subsection{Some Insights into the Problem}
Before proceeding any further, we first give some insights into the considered optimization problem. One can easily verify that the logarithmic function is monotonic. Therefore, we derive a simpler form of the objective in (OP) without affecting optimality as
\begin{eqnarray}
	\begin{aligned}[b]
		C\propto   {\log _2}\left( {{\gamma }} \right) + {\log _2}\left( {\prod\limits_{k \in {\cal K}} {{\lambda_k}} } \right).
	\end{aligned}
\label{eq20}\end{eqnarray}
Clearly, the first term in \eqref{eq20} is related to the average received SNR, which is independent of the antenna deployment parameters; and the second term involves the eigenvalues of the channel gain matrix, which is determining by the AAV node deployment. In view of this, we first seek to maximize the product of the eigenvalues of the channel gain matrix by optimizing the AAV node deployment; next, with given virtual MIMO channel matrix, we design the precoding and beamforming to meet the prescribed received SNR requirements  for confronting the hybrid Eves. Without loss of generality, signal paths exhibit high correlation in pure LoS-MIMO conditions, which leads to a lower rank of transmission channel matrix. Its core principle of MIMO system for improving capacity lies in leveraging spatial DoFs and the geometric design of antenna arrays. Hence, the pure LoS-MIMO channel matrix becomes high rank by meticulous design of the AAV node deployment to explore the synthetic aperture array’s potentials.

As previously analyzed, the \text{(OP)} boils down to optimizing AAV node deployment, whereby the product of eigenvalues of the channel gain matrix is maximized, i.e.,
\begin{align}
	\text{(P1)}: \quad \mathop {\max }\limits_{ \boldsymbol{\delta }} ~ \prod\limits_{k \in {\cal K}} {{\lambda _k}}.
	\label{eq21}\end{align}
Notice that the singular values of ${\bf{H}}$ are equivalent to those of ${\bf{\tilde H}}$ since both ${{\bf{B}}_G}$ and ${{\bf{B}}_A}$ are unitary by definition. In other words, replacing ${\bf{G}}$ with ${\bf{\tilde G}}\buildrel \Delta \over = {\bf{\tilde H}}^H{{\bf{\tilde H}}}$ offers a more straightforward way to derive the eigenvalues.

To facilitate the design of transmit precoding and receive beamforming, we apply the singular value decomposition (SVD) on the channel matrix ${\bf{\tilde H}}$ to precise control of received SNR, which is given by
\begin{eqnarray}
\begin{aligned}[b]
{\bf{\tilde H}} = \left[ {\begin{array}{*{20}{c}}
		{{{\bf{U}}_1}}&{{{\bf{U}}_0}}
\end{array}} \right]\left[ {\begin{array}{*{20}{c}}
		{\bf{D}}&{\bf{0}}\\
		{\bf{0}}&{\bf{0}}
\end{array}} \right]{\left[ {\begin{array}{*{20}{c}}
			{{{\bf{V}}_1}}&{{{\bf{V}}_0}}
	\end{array}} \right]^H},
	\end{aligned}
\label{eq22}\end{eqnarray}
where ${\bf{D}} = {\rm{diag}}\left(\{\sqrt {\boldsymbol{\lambda }}\}\right)\in \mathbb{R}^{K \times K}$ is a diagonal matrix with nonzero singular values ${\boldsymbol{\lambda }}\buildrel \Delta \over=\left[\lambda_1,\lambda_2,...,\lambda_K\right]^T$, ${{\bf{U}}_1}$ and ${{\bf{V}}_1}$ are left and right singular vectors corresponding to nonzero singular values. The transmit and receive array response matrix (also known as steering matrix or spatial signature matrix) are, respectively, defined as
\begin{align}
	\left\{ \begin{array}{l}
		{{\bf{A}}_A} \buildrel \Delta \over= \rho ({R}){{\bf{V}}_1^H}{{\bf{B}}_A},\\
		{{\bf{A}}_G} \buildrel \Delta \over= {{\bf{B}}_G}{{\bf{U}}_1}.
	\end{array} \right.
	\label{eq23}\end{align}

As mentioned earlier, energy-efficient transmission design is of utmost importance for AAV-enabled communication systems since AAVs are typically onboard battery-powered. Toward this end, we develop a precoding strategy to minimize transmit power, subject to constraints on the received SNR guarantee, leakage tolerance, and per-node transmit power as
\begin{align}
	\text{(P2)}: \quad &\mathop {\min }\limits_{{\bf{u}}} ~ \left\| {\bf{u}} \right\|_2^2 \tag{24a} \label{eq24a}\\&
	{\rm{s.t.}}  ~  {{\bf{A}}_A}{\bf{u}} \ge  {\boldsymbol{\xi }}, \tag{24b} \label{eq24b}\\&
	\kern 14pt	 \mathop {\max }\limits_{\Delta {{\bf{h}}_{E,q}} \in {{\cal U} _q}} {\rm{SNR}}_{E,q} \le {\Gamma _{E}}, \quad \forall q \in \mathcal{Q},\tag{24c} \label{eq24c}\\&
	\kern 15pt	{\left[{\bf{u}}{{\bf{u}}^H}\right]_{l,l}} \le {P_{\rm max}}, \quad \forall l \in {\cal L}, \tag{24d} \label{eq24d}
\end{align}
where ${\boldsymbol{\xi }} = \sqrt {{\gamma  }{{\sigma }_G^2}/{{\boldsymbol{\lambda }}K^2}}$ denotes the target received amplitude for the purpose of meeting desired received SNR requirements. The objective in \eqref{eq24a} is try to minimize the transmit power to against passive Eves since they usually keep radio silent, so any information cannot obtained. Based on PHY security transmission theory, energy leakage is one of the major reasons of deteriorating the secrecy performance, and thus security-oriented beamforming contributes to secrecy performance improvement. Moreover, less transmit power not only can save the transmit power consumption, which is able to provide energy-efficient transmission, but can reduce the message leakage, which can achieve low probability of detection (LPD) secure transmission. The constraint in \eqref{eq24b} is to protect the AAV toward the BS link. Our design approach adheres to providing QoS guarantees for the ground BS. The received SNR serves as an effective metric for QoS, as it is a crucial determinant of the maximum achievable rate and the probability of error. ${\rm{SNR}}_{E,q}={| {{\bf{h}}_{E,q}^H{\bf{u}}}|_2^2}/{\sigma _E^2}$, and ${\Gamma _{E}}$ denotes the maximum SNR tolerance for successfully wiretapping at Eves. The constraint in \eqref{eq24c} is imposed such that the maximum received SNRs of the active Eves remain below the tolerable threshold for given uncertainty sets. In practice, we set $\gamma \gg {\Gamma _{E}}$ to ensure secure communication. ${P_{\rm max}}$ is the maximum transmit power of the AAV. We should mention that an AAV is equipped with its own power amplifier, whose operation needs to constrain the AAV's peak power within its linearity~\cite{8089411}. The constraint in \eqref{eq24d} sets this physical bound of the maximum transmit power at each node radiated by the AAV. Moreover, it can be found that the eigenvalue maximization of the channel gain matrix design in \text{(P1)} not only achieve the channel capacity maximization, but it is related to the target received amplitude, which can further reduce the transmit power consumption.

We should mention that transmit power minimization formulation effectively achieve LPD secure transmission. Besides, leveraging the high flexibility of AAVs, some LPI ways can be also applied to dynamically disrupt the signal received by Eves, such as random rotation offsets and dynamic aperture sizes. The integrated use of LPI (low leakage power) and LPD (dynamic transmission strategy) transmission scheme can effectively thwart the behavior of individual Eves and their collaboration, thereby enhancing secrecy performance.

\begin{remark} \label{remark03}
The transmit power minimization problem formulation under QoS guarantee and the max-min received SNR problem subject to meeting a bound on the transmit power are equivalent up to scaling in the case that all the received SNRs are equal~\cite{1634819}. That is, our transmit power minimization design is consistent with the channel capacity maximization.
\end{remark}

On the other hand, a received beamforming optimization design approach of jamming suppression is built up based on LCMV method~\cite{van2002optimum}, i.e., 
\begin{align}
	\text{(P3)}: \quad &\mathop {\min }\limits_{\bf{w}} ~ {{\bf{w}}^H}{\bf{R}}_y\bf{w}\tag{25a} \label{eq25a}\\&
	{\rm{s.t.}} ~ {{\bf{w}}^H}{{\bf{A}}_G} = {\bf{1}}_{K \times 1}^T, \tag{25b} \label{eq25b}
\end{align}
indicating that minimize the jamming-plus-noise power while maintaining a distortionless response in the desired data streams.

\subsection{Asymptotic Optimality of AAV Node Deployment}

Now, our goal is to develop an algorithm to solve \text{(P1)}. Prior to design of the AAV node deployment, we analyze that the asymptotic behavior of the channel matrix is tractable, which is presented in the following theorem.
\begin{theorem} \label{theorem01}
	When $\omega \to 0$, the $k$th largest eigenvalue of the channel gain matrix is asymptotically computed as$\footnote{Mathematically, $\omega $ is small also corresponds to the case of $\phi \to 0$,	where the channel model conforms to the asymptotic characterization.}$
\begin{align}
	\setcounter{equation}{25}
	{\lambda _k} \approx {\left[ {\frac{{r_{G,k}}{r_{A,k}}}{(k - 1)!}}\right]^2}{(\omega \cos \phi )^{2(k - 1)}},\quad\forall k \in {\cal K}, 
\label{eq26}\end{align}
where ${r_{G,k}}$ and ${r_{A,k}}$ represent the $k$th diagonal elements of the upper triangular matrices ${\bf{R}}_G$ and ${\bf{R}}_A$ given in~\eqref{eq74}.	
\end{theorem}

\begin{IEEEproof}
	Please see Appendix~\ref{app1}.
\end{IEEEproof}

Utilizing the result in~\eqref{eq26}, we rewrite the objective of \text{(P1)} as
\begin{eqnarray}
	\begin{aligned}[b]
		\prod\limits_{k \in {\cal K}} {{\lambda _k}}& \approx \prod\limits_{k \in {\cal K}}{\left[{\frac{{r_{G,k}}{r_{A,k}}}{(k - 1)!}}\right]^2}{(\omega \cos \phi )^{2(k - 1)}}\\&
		=\prod\limits_{k \in {\cal K}}r_{A,k}^2\prod\limits_{k \in {\cal K}}{\left[{{{\frac{{r_{G,k}(\omega \cos \phi )}}{(k - 1)!}}^{k - 1}}}\right]^2}.
	\end{aligned}
	\label{eq27}\end{eqnarray}
As a result, \text{(P1)} is equivalently converted into a more compact form as
\begin{align}
	\text{(P1.1)}: \quad \mathop {\max }\limits_{{\boldsymbol{\delta }}}~\prod\limits_{k \in {\cal K}}{r_{A,k}^2}.
	\label{eq28}\end{align}
To obtain the optimal AAV node deployment, it is necessary to show each $\{r_{A,k}\}$ as an explicit function of $\{\delta_l\}$ via the following theorem.
\begin{theorem} \label{theorem02}
The relationship between ${\{r_{A,k}\}_{k \in {\cal K}}}$ and ${{\{ {\delta_l}\} }_{l \in {\cal L}}}$ can be expressed in a closed form as
\begin{align}
	\left\{ \begin{array}{l}
		{r_{A,1}^{2}} =  L, \\
		{r_{A,k}^{2}} = {\frac{{\sum\limits_{{\cal C}_k} {\prod\limits_{\{i < j\}\in {{\cal C}_k}} {({\delta _j} - {\delta _i})}^2} }}{{\sum\limits_{{\cal C}_{k - 1}} {\prod\limits_{\{i < j\}\in  {{\cal C}_{k-1}}} {({\delta _j} - {\delta _i})}^2}} }}
	\end{array} \right.,\quad\forall k > 1,
	\label{eq29}\end{align}
where ${{\cal C}_k}$ denotes the subsets containing all $k$-combinations of ${\cal L}$, and there are $\binom{L}{k}$ of them.
\end{theorem}

\begin{IEEEproof}
	Please see Appendix~\ref{app2}.
\end{IEEEproof}	
By applying the Theorem~\ref{theorem02}, we recast the objective function in \text{(P1.1)} as
\begin{eqnarray}
	\begin{aligned}[b]
		\prod\limits_{k \in {\cal K}} {r_{A,k}^2} & = L\prod\limits_{k \in {\cal K}\atop k \ne 1} {\frac{{\sum\limits_{{{\cal C}_k}} {\prod\limits_{\{i < j\}\in {{\cal C}_k}} {({\delta _j} - {\delta _i})}^2}}}{{\sum\limits_{{\cal C}_{k - 1}} {\prod\limits_{\{i < j\}\in {{\cal C}_{k-1}}} {({\delta _j} - {\delta _i})}^2}}}} \\&
		= \sum\limits_{{\cal C}_K} {\prod\limits_{\{i < j\} \in {{\cal C}_K}} {{({\delta _j} - {\delta _i})}^2}}.
	\end{aligned}
	\label{eq30}\end{eqnarray}
Then, \text{(P1.1)} is reformulated as
\begin{align}
	\text{(P1.2)}:\quad \mathop {\max }\limits_{{\boldsymbol{\delta}}}~{\mathscr{D}_{{\cal C}_K} ({\boldsymbol{\delta}})},
	\label{eq31}\end{align}
where
\begin{align}
	{\mathscr{D}}_{{\cal C}_K}({\boldsymbol{\delta}})= \sum\limits_{{\cal C}_K} {\prod\limits_{\{i < j\} \in {\cal C}_K} {({\delta _j} - {\delta _i})^2}}. 
	\label{eq32}\end{align}	
It is worth noting that the number of AAV nodes should be no less than that of parallel data streams. Consequently, we analyze two scenarios: Case 1: $K = L$ and Case 2: $K < L$.

If $K = L$, the objective of \text{(P1.2)} reduces to
\begin{align}
	{\mathscr{D}}_{{\cal C}_K}({\boldsymbol{\delta}})={\prod\limits_{\{i < j\}\in{{\cal C}_K}} {({\delta _j} - {\delta _i})}^2}. 
	\label{eq33}\end{align}
It is observed that the function ${\mathscr{D}}_{{\cal C}_K}({\boldsymbol{\delta}})$ reveals the squared determinant of the Vandermonde matrix that involves the AAV node deployment. This kind of Vandermonde determinant maximization problem over the interval $\{\delta_l\}_{l \in {\cal L}} \in [-1,1]$ was initially raised in~\cite{Koeffizienten}. Applying the results in~\cite{Fekete,Variables}, the optimal solution follows Fekete points or Gauss-Lobatto points.

If $K < L$, the following corollary is presented to obtain the optimal solution.

\begin{corollary} \label{corol01}
Let us divide $L$ AAVs into $K$ equal-size groups, and the $k$th group takes the value as
\begin{align}
	\{{\delta _l}\} = \{{\mu ^{(K)} _k}\},\quad {\rm{if}} \ k - 1 < lK/L \le k,
	\label{eq34}\end{align}
where $\{{\mu ^{(K)} _k}\}$ is used to represent a set of Fekete points with the length of $K$.
\end{corollary}

\begin{IEEEproof}
	Please see Appendix~\ref{app3}.
\end{IEEEproof}	

More specifically, all the $L$ AAV nodes are divided into $K$ groups with equal size. The AAV nodes falling each group are arranged in compactly co-located distribution, such as forming a virtual ULA with half-wavelength spacing between nodes, which enhances capacity through array gain, and the centers of these $K$ groups follow the Fekete-point distribution. This deployment strategy based on group-wise can be intuitively interpreted in terms of the lens of spatial division. To attain a spatial gain of $K$, only $K$ distinct eigenmodes are required to support $K$ spatially independent data streams, making the remaining ones redundant. Hence, we can already assure the existence of $K$ distinct eigenmodes via organizing all the AAV nodes into $K$ compact groups. We also investigate the NULA design in a non-asymptotic scenario through numerical optimizations.

\begin{remark} \label{remark04}
The previous derivation shows an interest fact that the maximum eigenvalue product of ${\bf{\tilde G}}$ involves variables unrelated to the time and space, which implies that the asymptotically optimal AAV node deployment remains invariant with respect to flight parameters. This fact contributes to the independent design of the AAV flight node deployment.  
\end{remark}

\subsection{Transmit Precoding Design}
Given the optimal AAV node deployment obtained in \text{(P1)}, the transmit and received array response vector can be determined in closed form using~\eqref{eq23}. Subsequently, let us turn to solving \text{(P2)}. Obviously, the constraint in~\eqref{eq24c} is convex, yet it contains semi-infinite variables owing to the uncertainty. To obtain a more tractable term, we transform the constraint~\eqref{eq24c} into a LMI by applying following lemma.

\begin{lemma} [S-Procedure~\cite{Convex}]\label{lemma01}
Introduce a function ${\mathscr{H}}_i\left( {\bf{x}} \right)$, $i \in \{ 1,2\}$, satisfying
\begin{align}
	{\mathscr{H}}_i\left( {\bf{x}} \right)\buildrel \Delta \over = {{\bf{x}}^H}{{\bf{A}}_i}{\bf{x}} + 2{\mathop{{\mathfrak{Re}}}\nolimits} \left\{ {{\bf{b}}_i^H{\bf{x}}} \right\} + {c_i},
	\label{eq35}\end{align}
where ${{\bf{A}}_i} \in \mathbb{H}^{L}$, ${\bf{b}}_i \in \mathbb{C}^{L \times 1}$, and ${c_i} \in \mathbb{R}$. The implication ${\mathscr{H}}_1\left( {\bf{x}} \right) \le 0 \Rightarrow {\mathscr{H}}_2\left( {\bf{x}} \right)\le 0$ holds if and only if there exists a variable $ \kappa >0$ so as to
\begin{eqnarray}
\begin{aligned}[b]
	\kappa \left[ {\begin{array}{*{20}{c}}
			{{{\bf{A}}_1}}&{{{\bf{b}}_1}}\\
			{{\bf{b}}_1^H}&{{c_1}}
	\end{array}} \right] - \left[ {\begin{array}{*{20}{c}}
			{{{\bf{A}}_2}}&{{{\bf{b}}_2}}\\
			{{\bf{b}}_2^H}&{{c_2}}
	\end{array}} \right]\succeq \boldsymbol{0}.
\end{aligned}
\label{eq36}\end{eqnarray}
Suppose there exists a point ${{\bf{x^*}}}$ yielding ${\mathscr{H}}_i\left( {\bf{x^*}} \right)<0$. 
\end{lemma}

Recall the uncertain sets in \eqref{eq17} as
\begin{align}
	{\mathscr{H}}_1\left(\Delta {{\bf{h}}_{E,q}}\right)=\Delta {{\bf{h}}^H_{E,q}}\Delta {{\bf{h}}_{E,q}} - \epsilon _{E,q}^2 \le 0, \quad \forall q \in \mathcal{Q}.
	\label{eq37}\end{align}
Inserting ${{\bf{h}}_{E,q}} = {{\bf{\tilde h}}_{E,q}} + \Delta {{\bf{h}}_{E,q}}$ into constraint \eqref{eq24c}, we have
\begin{align}
	{\mathscr{H}}_2\left(\Delta {{\bf{h}}_{E,q}}\right)&=\Delta {\bf{h}}_{E,q}^H\left( {\frac{{\bf{U}}}{\Gamma _{E}}} \right)\Delta {{{\bf{h}}}_{E,q}}\notag\\&
	+ 2{\mathop{{\mathfrak{Re}}}\nolimits} \left\{ {{\bf{\tilde h}}_{E,q}^H\left( {\frac{{\bf{U}}}{\Gamma _{E}}} \right)\Delta {{\bf{h}}_{E,q}}} \right\} \notag\\&
	+ {\bf{\tilde h}}_{E,q}^H\left( {\frac{{\bf{U}}}{\Gamma _{E}}} \right){{\bf{\tilde h}}_{E,q}} - \sigma _{E}^2, \quad \forall q \in \mathcal{Q},
	\label{eq38}\end{align}
where ${\bf{U}} \buildrel \Delta \over = {\bf{u}}{{\bf{u}}^H}$. By applying Lemma~\ref{lemma01} and introducing auxiliary variables ${\kappa _q} \ge 0$, ${\forall q \in \mathcal{Q}}$, to make the implication ${\mathscr{H}}_2\left(\Delta {{\bf{h}}_{E,q}}\right) \le 0$, the following LMI constraints should hold
\begin{align}
	\begin{array}{l}
		{{\bf{S}}_q}\left( {{{\bf{U}}},{\kappa _q}} \right)\\
		~	= \left[ {\begin{array}{*{20}{c}}
				{{\kappa _q}{{\bf{I}}_L}}&{\boldsymbol 0}_{L \times 1}\\
				{\boldsymbol 0}^T_{L \times 1}&{ \!-\! {\kappa _q}\epsilon _{E,q}^2\!+\!\sigma _{E}^2}
		\end{array}} \right]
		\!- \!\frac{1}{{{\Gamma _{E}}}}{\bf{\Xi}}_{q}^H {{{\bf{U}}}} {{\bf{\Xi}}_{q}} \succeq \boldsymbol{0}, \quad\forall q \in  \mathcal{Q},
	\end{array}
	\label{eq39}\end{align}
where ${{\bf{\Xi }}_{q}} = [ {{{\bf{I}}_L},{{{\bf{\tilde h}}}_{E,q}}}]$. Now constraint~\eqref{eq24c} involves only a finite number of terms, which is tractable for deriving the optimal solution. Based on this, a standard way to solve \text{(P2)} is to convert it into the following problem
\begin{align}
	\text{(P2.1)}: \quad &\mathop {\min }\limits_{\{{\bf{U}},{\boldsymbol{\kappa}}\}} ~ {\rm Tr}\left({\bf U}\right) \tag{40a} \label{eq40a}\\&
	{\rm{s.t.}}  ~  {\rm Tr}\left({{\bf \Pi} _k}{\bf{U}}\right)  \ge  {\xi _k^2},\quad \forall k \in {\cal K}, \tag{40b} \label{eq40b}\\&
	\kern 17pt	 	{{\bf{S}}_q}\left( {{{\bf{U}}},{\kappa _q}} \right) \succeq \boldsymbol{0}, \quad \forall q \in \mathcal{Q},\tag{40c} \label{eq40c}\\&
	\kern 17pt	{\rm Tr}\left({{\bf E}^{(l)}{\bf U}}\right) \le {P_{\rm max}}, \quad \forall l \in {\cal L}, \tag{40d}	\label{eq40d}\\&
	\kern 17pt  {\rm Rank}\left({\bf U}\right)=1, ~{\bf U} \succeq \boldsymbol{0}, \tag{40e} \label{eq40e}
\end{align}
where ${\boldsymbol{\kappa}} \buildrel \Delta \over = \{{\kappa _q}\}_{q \in {\cal Q}}$, ${\bf \Pi}_{k}={\bf a}_{A,k}^H{\bf a}_{A,k}$ with ${\bf a}_{A,k}$ is the $k$th row of the transmit array response matrix, i.e., ${\bf a}_{A,k}={\left[ {{{\bf{A}}_A}} \right]_{k,:}}$, $\xi _k$ denotes the $k$th element in $ {\boldsymbol{\xi }}$, and ${{\bf{E}}^{(l)}}  = {\bf{e}}_l{\bf{e}}_l^H$ with ${\bf{e}}_l=[\boldsymbol{0}_{(l-1) \times 1}^T,1,\boldsymbol{0}_{(L-l)\times 1}^T]^T\in {\mathbb{R}^{ L  \times 1}}$. Now, ${\rm Rank}\left({\bf U}\right)=1$ remains an obstacle to solving the problem. 

By relaxing constraint~\eqref{eq40e} to obtain a rank-relaxation version of \text{(P2.1)}, the resulting problem reduces to a convex SDP, which can be efficiently solved by SeDuMi~\cite{sturm1999using}.

We should mention that when the optimal solution ${\bf U}^*$ of the relaxed problem admits a rank-one matrix its principal component will be the optimal solution to the optimization problem. However, the matrix ${\bf U}^*$ obtained by solving the SDP in P2.1 will not be rank-one in general due to the relaxation. At this time ${\rm Tr}\left({\bf U}^*\right)$  is a lower bound on the power needed to satisfy the constraints for the dropping of rank constraint. Utilizing randomization~\cite{tseng2003further,zhang2000quadratic}, several studies have explored the generation of good solution of the optimization problem as follows%
		\begin{itemize}
			\item In the first method, we calculate the eigen-decomposition of ${\bf U}^*={\bf{B\Sigma }}{{\bf{B}}^{H}}$ and choose candidate precoding vectors ${\bf u}^*$ such that ${\bf u}^*={\bf{B\Sigma^{1/2} }}{\boldsymbol{\varsigma}}$, where ${\boldsymbol{\varsigma}}$ is uniformly distributed on the unit sphere.	This ensures that $ {{\bf{u}}^*}^H{\bf{u}}^*  = {\rm Tr}({\bf U}^*)$, irrespective of the particular realization of ${\boldsymbol{\varsigma}}$. 
			\item In the second method, inspired by Tseng ~\cite{tseng2003further}, we choose candidate precoding elements ${\left[ {{{\bf{u}}^*}} \right]_i} = \sqrt {{{\left[ {{{\bf{U}}^*}} \right]}_{i,i}}} {e^{j{\varsigma _i}}}$, where $\{\varsigma _i\}$ are independent and uniformly distributed on $[0,2\pi )$. This randomization ensures that  $\left| {[ {{{\bf{u}}^*}} ]_i}\right|^2={{\left[ {{{\bf{U}}^*}} \right]}_{i,i}}$.
			\item The third method, motivated by successful applications in related quadratically constrained quadratic programming problems~\cite{1634819}, uses  ${\bf u}^*={\bf{B\Sigma^{1/2} }}{\boldsymbol{\varsigma}'}$, where ${\boldsymbol{\varsigma}'}$ is a vector of zero-mean, unit-variance complex circularly symmetric uncorrelated Gaussian random variables. This ensures that $\mathbb{E}\{ {\bf{u}}^*{ {\bf{u}}^*}^H\}={\bf U}^*$. 
		\end{itemize}
By using ${\bf U}^*$ to generate a set of candidate precoding vectors, from which the “best” solution will be selected.

\subsection{Received Beamvector Design}
According to the array processing~\cite[Ch. 6.7]{van2002optimum}, the optimal solution of \text{(P3)} is computed as
\begin{eqnarray}
	\setcounter{equation}{41}
	\begin{aligned}[b]
		{{\bf{w}}^*} = {\bf{R}}_y^{ - 1}{{\bf{A}}_G}({\bf{A}}_G^H{\bf{R}}_y^{ - 1}{{\bf{A}}_G})^{-1}{{\bf{1}}_{L \times 1}}.
	\end{aligned}
	\label{eq41}\end{eqnarray}
In practical scenarios, the covariance matrix \(\mathbf{R}_y\) is usually unavailable, and thus the sample covariance matrix is typically employed instead, which is computed as
\begin{eqnarray}
	\begin{aligned}[b]
		{{{\bf{\hat R}}}_y} = \frac{1}{X}\sum\limits_{t \in {\cal X}} {{\bf{y}}(t){{\bf{y}}^H}(t)},
	\end{aligned}
	\label{eq42}\end{eqnarray}
where $X$ is the length of snapshots, and ${\cal X} \buildrel \Delta \over = [1,2,...,X]$.

\section{Some Extensions}\label{sec:Exte}
In prior studies, assumptions like linear arrays and perfect beamforming were used to simplify the model. In what follows, we extend the approach to more practical scenarios, i.e., high-dimensional virtual arrays and robust received beamforming.
\subsection{High-Dimensional Virtual MIMO Systems}
\begin{figure}[!t]
	\centering
	\includegraphics[width=0.75\columnwidth]{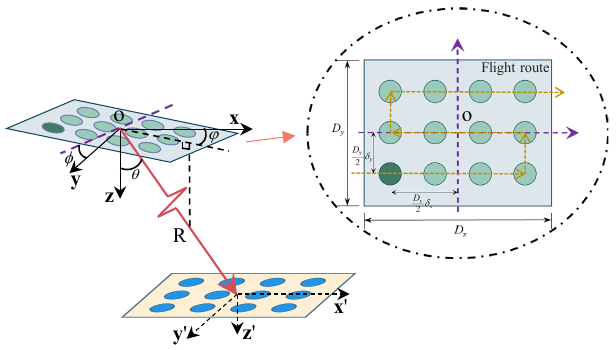}
	\caption{Model of high-dimensional virtual MIMO communications}
	\label{fig_3}
\end{figure}

A multidimensional virtual transmit array is easily constructed by a mobile AAV, which forms a synthetic aperture beamforming with higher orientation flexibility and better power concentration. 

Actually, two-dimensional (2D) planar arrays can be treated each axis as a linear array. Specifically, a virtual non-uniform planar array (NUPA) arranged by a single-antenna mobile AAV, which consists of $L$ flight nodes, i.e., $L = L_x \times L_y$ with aperture sizes $D_x$ and $D_y$. The normalized spacing vectors of the virtual NUPA relative to the center along two axes are denoted by ${{\boldsymbol{\delta }}_x}=\{{\delta }_{x,l_x}\}\in \mathbb{R} ^{L_x \times 1}$ and ${{\boldsymbol{\delta }}_y}=\{{\delta }_{y,l_y}\}\in \mathbb{R} ^{L_y \times 1}$, respectively, $ l_x\in {\cal L}_x \buildrel \Delta \over =[1,2,...,L_x]$ and $ l_y\in {\cal L}_y \buildrel \Delta \over =[1,2,...,L_y]$. The BS is equipped with an $N$-elements uniform planar array (UPA), i.e., $N = N_x \times N_y$. Utilizing 3D geometric coordinates, as depicted in Fig.~\ref{fig_3}, a far-field LoS-MIMO channel is established. The origin is set at the center of the virtual NUPA. The $x$-axis and $y$-axis align with the column-wise and row-wise elements of the ground-based UPA, respectively, while the $z$-axis is oriented towards the ground. The rotation offset $\phi$ occurs around the $z$-axis. Thereafter, the coordinates of each transmit AAV node and received element are determined as
\begin{align}
	\left\{ {\begin{array}{*{20}{l}}
			{{x_{A,l_x,l_y}} = \frac{{{D_x}{\delta _{x,l_x}}\cos \phi  - {D_y}{\delta _{y,l_y}}\sin \phi }}{2}},\\
			{{y_{A,l_x,l_y}} = \frac{{{D_x}{\delta _{x,l_x}}\sin \phi  + {D_y}{\delta _{y,l_y}}\cos \phi }}{2}},\\
			{z_{A,l_x,l_y} = 0},
	\end{array}} \right.
	\label{eq43}\end{align}
and 
\begin{align} 
	\left\{ {\begin{array}{*{20}{l}}
			{{x_{G,{n_x,n_y}}} = \frac{{2{n_x} - 1 - {N_x}}}{2}d_x + R\sin \theta \cos \varphi ,}\\
			{{y_{G,{n_x,n_y}}} = \frac{{2{n_y} - 1 - {N_y}}}{2}d_y + R\sin \theta \sin \varphi ,}\\
			{{z_{G,{n_x,n_y}}} = R\cos \theta},
	\end{array}} \right.
	\label{eq44}\end{align}
where $d_x$ and $d_y$ denotes the inter-element spacing of UPA, respectively, $ n_x\in {\cal N}_x \buildrel \Delta \over =[1,2,...,N_x]$, and $ n_y\in {\cal N}_y \buildrel \Delta \over =[1,2,...,N_y]$. 

In a way similar to the previous section, we can obtain the radio wave propagation range based on the coordinates, and further calculate the equivalent MIMO channel matrix. Following a similar derivation as virtual transmit linear array, we can prove that the optimal high-dimensional AAV node deployment follows the 2D Fekete distributions~\cite{BRIANI20122477}. 

Next, based on the given high-dimensional AAV node deployment, we construct a transmit power minimization problem for secure communications and utilize the LCMV technique for jamming suppression.

\subsection{Robust Adaptive Beamforming Design}
Due to the fact that the AAV inevitably experiences wind-induced body jitters and route deviations, and positioning modules introduce measurement errors caused by its limited accuracy, the transmission link suffers from performance degradation. On account of this, we study a robust adaptive beamforming method to enhance the performance. 

By means of the bounded location error model, the received channel can be expressed as
\begin{align}
&{{\bf{a}}_{G,k}} = {{{\bf{\tilde a}}}_{G,k}} + \Delta {{\bf{a}}_{G,k}},\quad \forall k \in {\cal K},\label{eq45}\\&
		{\cal A}_k = \left\{ {\Delta {{\bf{a}}_{G,k}} \in {\mathbb{C}^{N \times 1}}:\Delta {\bf{a}}_{G,k}^H\Delta {{\bf{a}}_{G,k}} \le \epsilon _{G,k}^2} \right\},\quad\forall k \in {\cal K},\label{eq46}
\end{align}
where ${\bf a}_{G,k} ={\left[ {{{\bf{A}}_G}} \right]_{:,k}}$ is the $k$th column of the received array response matrix, ${{{\bf{\tilde a}}}_{G,k}}$ and $\Delta {{\bf{a}}_{G,k}}$ denote corresponding the estimate channel and the uncertainty, respectively, and ${\cal A}_k $ represents the uncertainty sets. Then the robust adaptive beamforming \text{(P3)} can be rewritten as
\begin{align}
	\text{(P3.1)}: \quad &\mathop {\min }\limits_{\bf{w}} ~ {{\bf{w}}^H}{\bf{R}}_y\bf{w} \tag{47a} \label{eq47a}\\&
	{\rm{s.t.}} ~ \mathop {\min }\limits_{\Delta{\bf{a}}_{G,k}\in {\cal A}_k} ~ |{{\bf{w}}^H}{{\bf{a}}_{G,k}}| \ge {1}, \quad \forall k \in {\cal K}. \tag{47b} \label{eq47b}
\end{align}
Resorting to the worst-case signal-to-jamming-plus-noise ratio (SJNR) maximization method~\cite{10027476}, the \text{(P3.1)} is equivalent to following max-min fairness problem
\begin{align}
	\text{(P3.2)}: \quad & \mathop {\max}_{\bf{w}}  ~  \mathop {\min }\limits_{k\in{\cal K}} ~\mathop {\min }\limits_{\{\Delta{\bf{a}}_{G,k}\in {\cal A}_k\}} ~ \frac{|{{\bf{w}}^H}{{\bf{a}}_{G,k}}|^2}{{{{\bf{w}}^H}{{\bf{R}}_y}{\bf{w}}}}.\tag{48} \label{eq48}
\end{align}
To solve \text{(P3.2)}, the worst-case (smallest) received SJNR over the uncertainty is defined as 
\begin{eqnarray}
	\begin{aligned}[b]
	\setcounter{equation}{49}
	{\eta _k} = \frac{\mathop {\min }\limits_{\{ \Delta {{\bf{a}}_{G,k}} \in {{\cal A}_k}\} }~|{{\bf{w}}^H}{{\bf{a}}_{G,k}}{|^2}}{{{\bf{w}}^H}{{\bf{R}}_y}{\bf{w}}},
	\end{aligned}
\label{eq49}\end{eqnarray}
By introducing an auxiliary variable $\varpi = \mathop {\min }\limits_{k\in{\cal K}}~ \{ {\eta _k}\}\in \mathbb{R}^{+}$, the robust max-min fairness problem \text{(P3.2)} is equivalently expressed as
\begin{align}
	\text{(P3.3)}: \quad &	\mathop {\max }\limits_{\bf{w}} ~\varpi \tag{50a} \label{eq50a}\\&
	 {\rm{s.t.}} ~{\eta _k}  \ge \varpi , \quad\forall k \in {\cal K}.\tag{50b} \label{eq50b}
\end{align}
In the same manner, we replace the non-convex constraint with a tractable convex deterministic constraint by utilizing S-Procedure, thereby facilitating the solutions of the problem.

\section{Simulation Results}\label{sec:Simu}
In this section, simulation results are presented under different parameter settings to evaluate the performance of our proposed AAV-enabled communication scheme. Unless specified otherwise, the system parameters are configured as follows. An AAV is dispatched to provide communications to the ground BS, with coordinate $(-22{\rm ~m}, 127{\rm ~m}, 75{\rm ~m})$, in the presence of $Q=2$ Eves, with coordinates $(-95{\rm ~m}, 88{\rm ~m}, 75{\rm ~m})$ and $(104{\rm ~m}, 78{\rm ~m}, 75{\rm ~m})$. The carrier frequency is $f_c=1 {\rm ~GHz}$. The maximum transmit power of the AAV is $P_{\rm max} = 0 {\rm ~dBm}$. For simplicity, we assume received thermal noise for both BS and Eves are identity, i.e., $10\lg (\sigma_G^2)=10\lg (\sigma_{E,q}^2)=-100$ \text{dBm}~\cite{8078202}. To guarantee secure communications, the maximum tolerable SNR for Eves is set to ${\Gamma _{E}} = 0 {\rm~dB}$. Compared with traditional single-antenna omnidirectional transmission, the proposed synthetic aperture beamforming exhibits directivity and thus provides spatial security, intuitively. For array element deployment and precoding designs, we compared the proposed scheme with the traditional ULA and zero-forcing (ZF) precoding schemes~\cite{9354156}.
\begin{figure}[tb]
	\begin{center}
		\includegraphics[width=1\columnwidth]{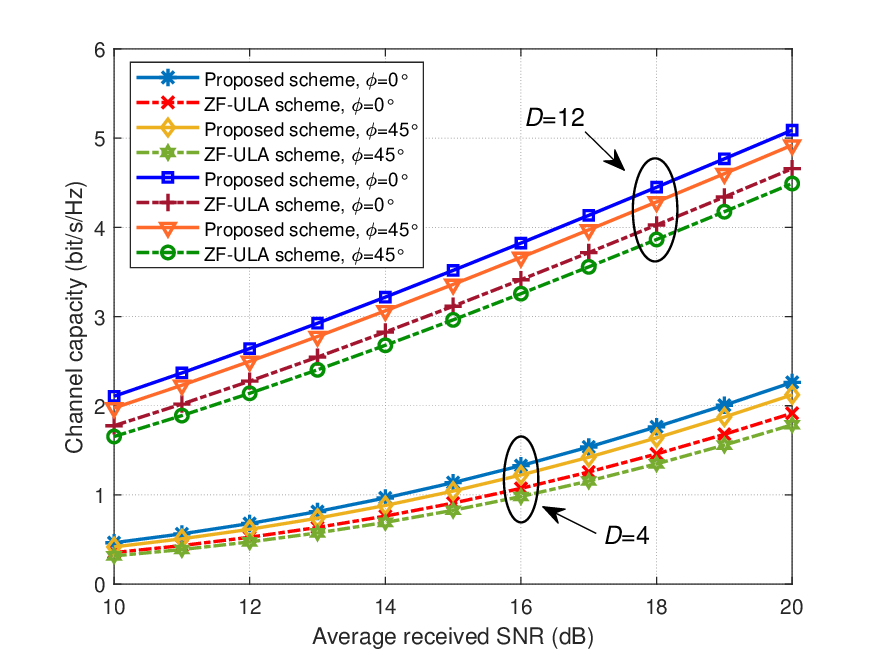}
	\end{center}
	\caption{Channel capacity versus average received SNR for different aperture sizes and rotation offsets.}
	\label{fig4}
\end{figure}

In Fig.~\ref{fig4}, we reveal the system channel capacities of the two approaches versus the transmit power. Typically, in the case of $L=16$ and $N=32$. It is observed that a distinct slope difference between the curves of ZF-ULA scheme and the proposed scheme under varying desired received SNRs, particularly in the high-SNR regimes. This demonstrates the superior channel capacity gains for our proposed scheme. This improvement arises from NULA's ability to maximize the product of the eigenvalues of the channel gain matrix and control the received SNR. In addition, deploying a larger aperture size of the virtual array enhances transmission capabilities. The results in Fig.~\ref{fig4} are consistent with~\eqref{eq26}. That is, the eigenvalues of the channel gain matrix are inversely related to the rotation offsets. Therefore, larger rotation offsets result in a reduction of the realized channel capacity. This is attributed to the fact that the realized channel capacity is corresponding to the projected equivalent array aperture. On the other hand, larger rotation offset leads to more intense dynamic scrambling, and thus reduce the possibility of signals being intercepted, yet lower channel capacity. The results also indicate that a higher channel gain can be achieved using the proposed optimized NULA design, even in the non-asymptotic scenario (i.e., $\omega$ as a finite value).
\begin{figure}[t]
	\begin{centering}
		\subfigure[]{\begin{centering}
				\includegraphics[width=1\columnwidth]{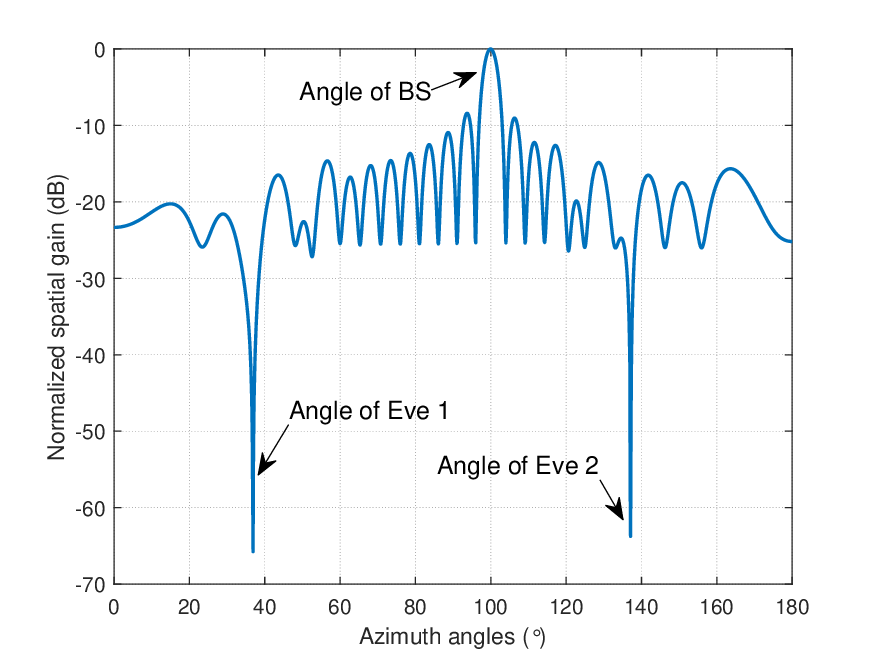}
				\par\end{centering}
		}
		\subfigure[]{\begin{centering}
				\includegraphics[width=1\columnwidth]{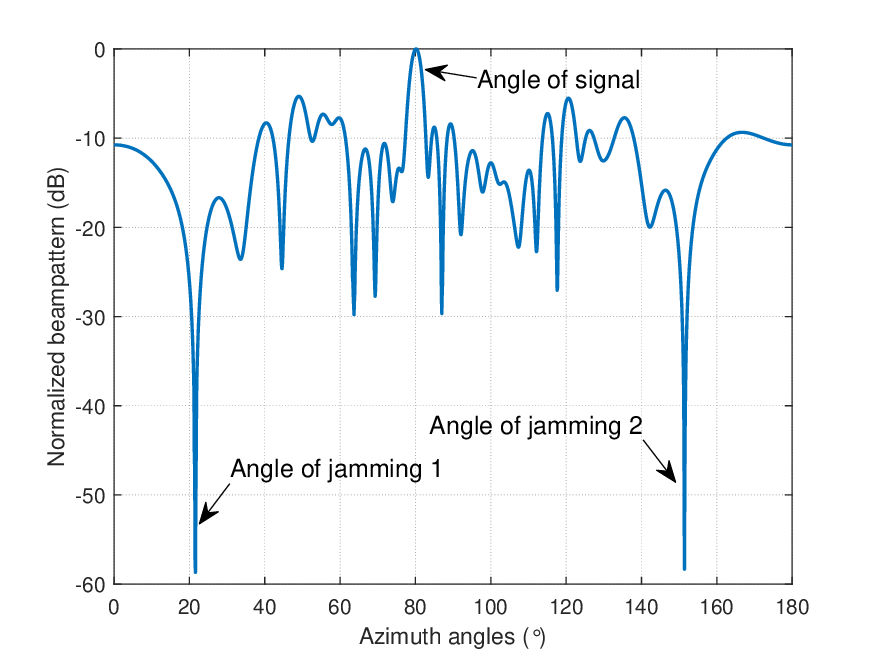}
				\par\end{centering}
		}
		\par\end{centering}
	\caption{Resultant beam responses in the azimuth angle dimension of the proposed method. (a) Transmit spatial gain. (b) Received beampattern.}
	\label{fig5}\end{figure}
	
In the next simulation, based on the coordinates of the BS and the Eves, we depict the typical resultant beam responses of our proposed scheme. We choose $L=24$ and $N=32$. Particularly, we plot the normalized transmit spatial gain in Fig.~\ref{fig5} (a), and the normalized received beampattern in Fig.~\ref{fig5} (b). A general observation from Fig.~\ref{fig5} (a) is that a sharp spatial gain peak is synthesized in the azimuth angle of the BS, which provides QoS assurance of the AAV-to-BS link. Besides, the spatial gains are so poor along other azimuth angles, especially along the azimuth angles of Eves, which effectively prevented eavesdropping. The performance gain is due to our synthetic aperture beamforming design and minimum transmission power criterion. As can be seen from Fig.~\ref{fig5} (b), two deep nulls are formed at the azimuth angles of the jamming, while the response remains zero along the desired angle. These results confirm the effective jamming suppression alongside distortionless response preservation toward the arrival angle of signals.	
\begin{figure}[tb]
	\begin{center}
		\includegraphics[width=1\columnwidth]{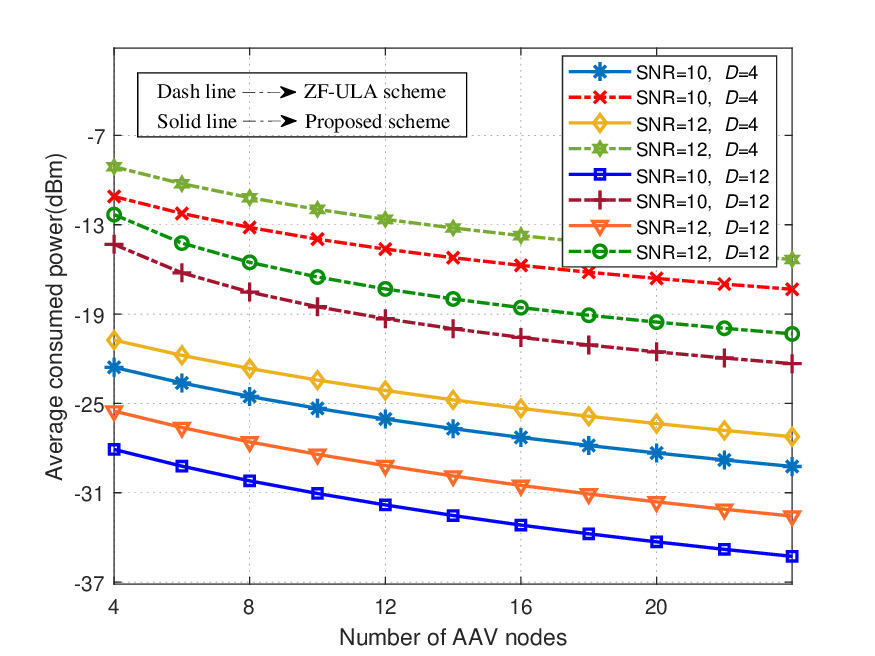}
	\end{center}
	\caption{Average consumed power versus number of AAV nodes for different aperture sizes and average received SNRs.}
	\label{fig6}
\end{figure}

To show the behavior of energy-efficient beamforming design, we list the transmit power consumption with different aperture sizes and average received SNRs in Fig.~\ref{fig6}. It can be observed that the power consumption decreases as more AAV nodes are deployed. This is mainly due to the fact that more AAV nodes are able to enhance array’s capability in array signal processing. In addition, we see that less power consumption for our proposed scheme compared to conventional ZF-ULA design. This is expected since the careful design of the AAV node deployment requires lower received power to meet the desired SNR and the precoding design minimizes the transmit power consumption. Moreover, a larger aperture size of the virtual array and a lower desired received SNR consume less transmit power. It should be mentioned that lower transmission power can reduce information leakage while meeting reception requirements, thereby enhancing security.
\begin{figure}[tb]
	\begin{center}
		\includegraphics[width=1\columnwidth]{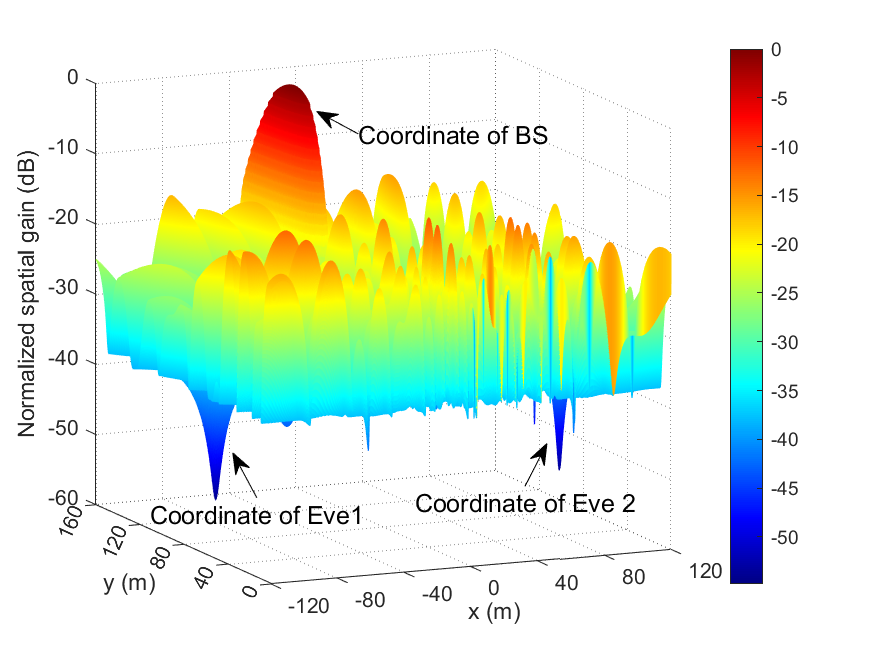}
	\end{center}
	\caption{Transmit spatial gain distribution of the proposed virtual NUPA synthesis method.}
	\label{fig7}
\end{figure}

We further show the radiation pattern of the proposed virtual NUPA synthesis method in Fig.~\ref{fig7}  at $L=6 \times 6$, $N=8 \times 8$, and $D=4 \times 4$. As seen, a sharp spatial gain peak is synthesized around the BS coordinate while maintaining uniformly low spatial gain elsewhere, especially in the coordinates of the Eves. That is, even if Eves can sample and combine correctly, the signal power along the sidelobes remains extremely low, effectively enabling LPD secure transmission. The results confirm our proposed scheme's dual capability to guarantee reliable A2G transmission and effectively preventing interception.
\begin{figure}[tb]
	\begin{center}
		\includegraphics[width=1\columnwidth]{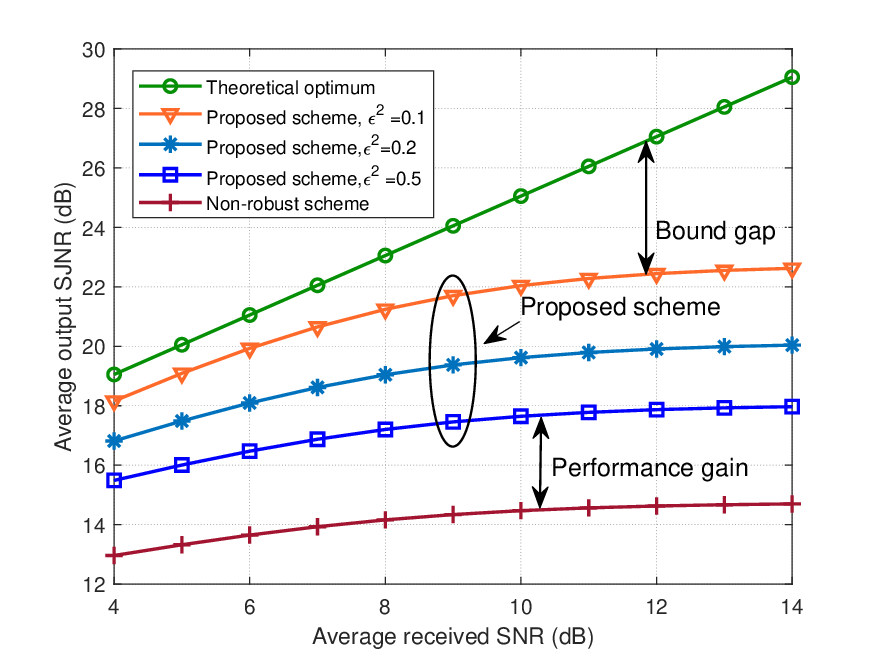}
	\end{center}
	\caption{Average output SJNR versus the received SNR for different channel uncertainty level.}
	\label{fig8}
\end{figure}

In the last scenario, we quantify the average output SJNR versus the received SNR for different channel uncertainty levels in Fig.~\ref{fig8}, where the theoretical upper bound is also presented. As expected, the non-robust method fails to achieve optimal performance under uncertainty, and the proposed robust method maintains a higher average output SJNR than the non-robust approach. Moreover, the average output SJNR exhibits a rapid increase under low received SNR, with a gradual slowdown of rising speed as the received SNR grows.

\section{Conclusion}\label{sec:Conc}
This paper studied an AAV-enabled secure communication problem against hybrid wiretapping, under which conventional approaches may fail to provide satisfactory performance due to load and power limitations of an AAV. To tackle this issue, we investigated a synthetic aperture beamforming design method formed by a single-antenna mobile AAV. Concretely, we integrated the design of virtual array deployment, transmission symbol strategy, transmit precoding, and received beamforming to maximize the system channel capacity. By providing some insights of the considered problem, the optimization problem can be divided into two independent problems. For the AAV node deployment, we asymptotically analyzed the expression of the channel eigenvalues and proved that the optimal virtual NULA deployment follows the Fekete-point distribution. Then, an energy-efficient method for minimizing the transmit power was studied to obtain the optimal precoding, while meeting the QoS for the BS, leakage tolerance of the Eves, and per-node transmit power constraints. Next, the received beamforming was designed via LCMV method to suppress jamming. Moreover, we revealed that the proposed scheme can be extended to high-dimensional arrays and robust adaptive beamforming. Finally, extensive simulations confirmed the superiority of AAV-enabled communications for performance enhancements. The proposed scheme holds broad application prospects in the near future. The virtual antenna array, which is constructed by a single-antenna mobile AAV to achieve synthetic aperture beamforming, represents an emerging technology. This research opens a way for flexible communication, significantly enhancing performance capabilities and coverage range. Potential application scenarios include emergency communications, coverage in remote areas, and military applications. Due to space constraints, this paper leaves several interesting challenges. Notably, strict synchronization, uncertain impacts, and computational complexity are inspiring work for advancing application.

\IEEEpubidadjcol 

\appendices
\section{Proof of Theorem~\ref{theorem01}} \label{app1}
Recalling \eqref{eq8} and performing Taylor expansion, we get
\begin{eqnarray}
	\begin{aligned}[b]
	\setcounter{equation}{51}
		{\tilde h_{n,l}}& = {e^{j\omega  \frac{{(2n - 1 - N)}}{N}{\delta _l}\cos \phi }},\\&
		= \sum\limits_{\upsilon  = 0}^\infty  {\frac{[j\omega {\delta _l}\cos \phi(2n - 1 - N)/N]^\upsilon }{\upsilon!}}.
	\end{aligned}
	\label{eq51}\end{eqnarray}
Based on the derivations in \eqref{eq51}, ${\bf{\tilde H}}$ can be decomposed as
\begin{align}
	{\bf{\tilde H}} = {{\bf{D }}_G}{\bf{\Lambda}}{\bf{D}}_A^T,
	\label{eq52}\end{align}
where ${\bf{\Lambda}}={\rm diag}(\alpha _1,\alpha _2,...)$ is a $\infty$-by-$\infty$ diagonal matrix with ${\alpha_\upsilon} = {(j\omega \cos \phi )^{\upsilon - 1}}/(\upsilon - 1)!, \upsilon \in [1,2,...]$, ${{\bf{D}}_A} \in {{\mathbb R}^{L \times \infty }}$ and ${{\bf{D}}_G} \in {{\mathbb R}^{N \times \infty}}$ are the Vandermonde matrices involved the element deployment, which satisfy
\begin{align}
	{{\bf{D}}_A} = {\left( {\begin{array}{*{20}{c}}
				1&{{\delta _1}}&{\delta _1^2}& \cdots \\
				1&{{\delta _2}}&{\delta _2^2}& \cdots \\
				\vdots & \vdots & \vdots & \ddots \\
				1&{{\delta _L}}&{\delta _L^2}& \ldots 
		\end{array}} \right)_{L \times \infty }},
	\label{eq53}\end{align}
and 
\begin{align}
	{{\bf{D}}_G} = {\left( {\begin{array}{*{20}{c}}
				1&{\frac{{1 - N}}{N}}&{{{(\frac{{1 - N}}{N})}^2}}& \cdots \\
				1&{\frac{{3 - N}}{N}}&{{{(\frac{{3 - N}}{N})}^2}}& \cdots \\
				\vdots & \vdots & \vdots & \ddots \\
				1&{\frac{{N - 1}}{N}}&{{{(\frac{{N - 1}}{N})}^2}}& \ldots 
		\end{array}} \right)_{N \times \infty }}.
	\label{eq54}\end{align}
To facilitate derivation, the matrices ${{\bf{D}}_G}$, ${{\bf{\Lambda}}}$, and ${{\bf{D}}_A}$ are operated with block decomposition as ${{\bf{D}}_G}=({{\bf{D}}_{G,1}},{{\bf{D}}_{G,2}},...)$, ${\bf{\Lambda}}={\rm diag}({\bf{\Lambda}}_{1},{\bf{\Lambda}}_{2},...)$, and ${{\bf{D}}_A}=({{\bf{D}}_{A,1}},{{\bf{D}}_{A,2}},...)$, where ${{\bf{D}}_{G,i}}\in {\mathbb{R}^{N \times N}}$, ${{\bf{\Lambda}}_{i}}\in{\mathbb{C}^{N \times N}}$, and ${{\bf{D}}_{A,i}}\in {\mathbb{R}^{L \times N}}$ are the $i$th sub-matrix of ${{\bf{D}}_G}$, ${{\bf{\Lambda}}}$, and ${{\bf{D}}_A}$, $\forall i$, respectively. Hence, \eqref{eq52} becomes
\begin{align}
	{\bf{\tilde H}} = \sum\limits_{i = 1}^\infty  {{{\bf{\tilde H}}}_i},
	\label{eq55}\end{align}
where ${{\bf{\tilde H}}_i} = {{\bf{D}}_{G,i}}{{\bf{\Lambda}}_i}{\bf{D}}_{A,i}^T$, $\forall i$. The primary sub-matrix ${{\bf{\tilde H}}_1}$ is then selected for approximation analysis. Defining ${{\bf{\tilde G }}_1} \buildrel \Delta \over = {{\bf{\tilde H}}_1^H}{\bf{\tilde H}}_1$ with eigenvalues ${\tilde \lambda _l }$, we derive
\begin{align}
	&\mathop {\lim }\limits_{\omega  \to 0} \frac{{\ln \left( {\prod\limits_{l \in {\cal L}} {\tilde \lambda }_l}  \right)}}{{\ln \left( {\omega \cos \phi } \right)}}\nonumber\\&
	\!	=\! \mathop {\lim }\limits_{\omega  \to 0} \frac{{\ln \left( {\det ({{{\bf{\tilde G }}}_1})} \right)}}{{\ln ({\omega \cos \phi })}}\nonumber\\&
	\!	=\! \mathop {\lim }\limits_{\omega  \to 0} \frac{{\ln \left( {\det ( {{{\bf{D}}_{G,1}}{{\bf{\Lambda}}_1}{\bf{D}}_{A,1}^T{{\bf{D}}_{A,1}}{\bf{\Lambda}}_1^H{\bf{D}}_{G,1}^T})} \right)}}{{\ln \left( {\omega \cos \phi } \right)}}\nonumber\\&
	\!	=\! \mathop {\lim }\limits_{\omega  \to 0} \frac{{\ln \left( {\det ( {{\bf{D}}_{G,1}}{\bf{D}}_{G,1}^T)\det({{\bf{\Lambda}}_1})\det({\bf{D}}_{A,1}^T{{\bf{D}}_{A,1}})\det({\bf{\Lambda}}_1^H)} \right)}}{{\ln \left( {\omega \cos \phi } \right)}}\nonumber\\&
	\!	=\! \mathop {\lim }\limits_{\omega  \to 0} \frac{{\ln \left( {\det ( {{\bf{D}}_{G,1}^T{{\bf{D}}_{G,1}}}   {{\bf{D}}_{A,1}^T{{\bf{D}}_{A,1}}}) \frac{{{{\left( {\omega \cos \phi } \right)}^{L(L\! -\! 1)}}}}{{{{\left[{\prod\limits_{l \in {\cal L}} {\left( {l\! -\! 1} \right)!} }\right]}^2}}}} \right)}}{{\ln \left( {\omega \cos \phi } \right)}}\nonumber\\&
	\!	=\! L(L - 1).
	\label{eq56}\end{align}
Let us employ the QR decomposition (also called the QR factorization) of Vandermonde matrix ${{\bf{D}}_{A,1}}$ as
\begin{align}
	{{\bf{D}}_{A,1}} = {{\bf{Q}}_{A,1}}{{\bf{R}}_{A,1}},
	\label{eq57}\end{align}
where ${{\bf{Q}}_{A,1}}\in {\mathbb{R}^{L \times L}}$ and ${{\bf{R}}_{A,1}}\in {\mathbb{R}^{L \times N}}$ are unitary and upper triangular matrices, respectively. We construct a matrix with the same eigenvalue  ${\tilde {\bf G}_1}$ as 
\begin{align}
	{\bf {\hat G}} \buildrel \Delta \over  = {\bf{Q}}_{A,1}^T{\tilde {\bf G}_1} {{\bf{Q}}_{A,1}}
	= {{\bf{R}}_{A,1}}{{\bf{\Lambda}}_1^H}{\bf{D}}_{G,1}^T{{\bf{D}}_{G,1}}{\bf{\Lambda}}_1{{\bf R}}_{A,1}^T.
	\label{eq58}\end{align}
The $l$th diagonal entry of ${\bf{\hat G}}$, denoted by ${{g}_l}$, leads to~\cite{7546944}
\begin{align}
	\mathop {\lim }\limits_{\omega  \to 0} \frac{{\ln {{g}_l}}}{{\ln \left( {\omega \cos \phi } \right)}} = 2(l - 1), \quad\forall l \in {\cal L}.
	\label{eq59}\end{align}
Since $\omega \to 0$, ${\ln (\omega\cos \phi )}$ tends toward $- \infty $ and thus, the bound slope of ${{\tilde \lambda }_l}$ satisfies
\begin{eqnarray}
	\begin{aligned}[b]
	\mathop {\lim }\limits_{\omega  \to 0} \frac{{\ln \left( {{\tilde \lambda }_l} \right)}}{{\ln \left( {\omega \cos \phi } \right)}}  & \ge \mathop {\lim }\limits_{\omega  \to 0} \frac{{\ln \left( {\sum\limits_{i \in  {\cal L}} {\tilde \lambda }_i} \right)}}{{\ln \left( {\omega \cos \phi } \right)}}\\&
	\overset{{(a)}}\ge \mathop {\lim }\limits_{\omega  \to 0} \frac{{\ln \left( {\sum\limits_{i \in  {\cal L}} {{g_i}} } \right)}}{{\ln \left( {\omega \cos \phi } \right)}}\\&
 	= 2(l - 1),
	\end{aligned}
	\label{eq60}\end{eqnarray}
where the inequality $(a)$ follows the Hermitian matrix majorization relations~\cite[Co. 4.3.34.]{Matrix}. Thus, it can be further obtained as 
\begin{eqnarray}
	\begin{aligned}[b]
	\mathop {\lim }\limits_{\omega  \to 0} \frac{{\ln \left( {\prod\limits_{l \in {\cal L}} {\tilde \lambda }_l} \right)}}{\ln \left( {\omega \cos \phi } \right)} &= \mathop {\lim }\limits_{\omega  \to 0} \sum\limits_{l \in {\cal L}} {\frac{\ln \left( {{\tilde \lambda }_l} \right)}{\ln \left( {\omega \cos \phi } \right)}}\\&  
	\ge \sum\limits_{l \in {\cal L}} {2(l - 1)}\\&
	 = L(L - 1).
	\end{aligned}
	\label{eq61}\end{eqnarray}
Recalling to~\eqref{eq56}, we get that the equality in~\eqref{eq60} and~\eqref{eq61} hold, i.e.,
\begin{align}
	\mathop {\lim }\limits_{\omega  \to 0} \frac{{\ln \left( {{\tilde \lambda }_l} \right)}}{\ln \left( {\omega \cos \phi } \right)} = 2(l - 1), \quad \forall l \in {\cal L}.
	\label{eq62}\end{align}

Actually, ${\bf{\tilde H}}$ can be perceived as a perturbation of ${{{\bf{\tilde H}}}_1}$ (cf. Eq.~\eqref{eq55}). Applying the Weyl’s Perturbation Theorem~\cite{Weyl,Perturbation}, yields
\begin{align}
	| {\sqrt {\lambda _l}  - \sqrt {\tilde \lambda }_l} | \le \| {\sum\limits_{i = 2}^\infty  {{\bf{\tilde H}}}_i} \|_2.
	\label{eq63}\end{align}
Inserting ${{\bf{\tilde H}}_i} = {{\bf{D}}_{G,i}}{{\bf{\Lambda}}_i}{\bf{D}}_{A,i}^T$ into the right-side of~\eqref{eq63}, we have
\begin{eqnarray}
	\begin{aligned}[b]
		&\mathop {\lim }\limits_{\omega  \to 0} \frac{\| {\sum\limits_{i = 2}^\infty  {{\bf{\tilde H}}}_i} \|_2}{{( \omega \cos \phi )}^{L - 1}}\\&
		\overset{{(b)}}\le \mathop {\lim }\limits_{\omega  \to 0} \frac{\sum\limits_{i = 2}^\infty  {\| {\bf{D}}_{G,i} \|_2}  \cdot \| {{\bf{\Lambda}}_i} \|_2 \cdot \| {{\bf{D}}_{A,i}^T} \|_2}{{( \omega \cos \phi )}^{L - 1}}\\&
		\overset{{(c)}}= \mathop {\lim }\limits_{\omega  \to 0} \frac{\sum\limits_{i = 2}^\infty  {\| {\bf{D}}_{G,i} \|_2}  \cdot \frac{{{( \omega \cos \phi )} ^{L(l - 1)}}}{[ {L(l - 1)}]!} \cdot \| {{\bf{D}}_{A,i}^T} \|_2}{{( \omega \cos \phi )}^{L - 1}}= 0,
	\end{aligned}
	\label{eq64}\end{eqnarray}
where the inequality $(b)$ is obtained due to the Triangle inequality, and the equality $(c)$ holds based on the fact that spectral norm of a matrix is its largest singular value~\cite[Th. 5.6]{Matrix}. Utilizing~\eqref{eq63} and~\eqref{eq64}, we obtain
\begin{eqnarray}
	\begin{aligned}[b]
		&\mathop {\lim }\limits_{\omega  \to 0} \left| {\sqrt {\frac{{{\lambda _l}}}{{{{(\omega \cos \phi )}^{2(l - 1)}}}}}  - \sqrt {\frac{{{{\tilde \lambda }_l}}}{{{{(\omega \cos \phi )}^{2(l - 1)}}}}} } \right|\\&
		\le \mathop {\lim }\limits_{\omega  \to 0} \frac{{\| {\sum\limits_{i = 2}^\infty  {{\bf{\tilde H}}}_i} \|_2}}{{\left( {\omega \cos \phi } \right)}^{L - 1}}= 0,
	\end{aligned}
	\label{eq65}\end{eqnarray}
which shows that ${\lambda _l}$ equals to ${{\tilde \lambda }_l}$ with the high-order infinitesimals ${(\omega \cos \phi)}^{2(l - 1)}$. In consequence, we have
\begin{align}
	\mathop {\lim }\limits_{\omega  \to 0} \frac{{\ln ( \lambda _l )}}{\ln ( \omega \cos \phi )} = \mathop {\lim }\limits_{\omega  \to 0} \frac{{\ln ( {{{\tilde \lambda }_l}} )}}{\ln (\omega \cos \phi )} = 2(l - 1), \quad\forall l \in {\cal L}.
	\label{eq66}\end{align}
	
Denote by $\{{\boldsymbol{\chi }}_l\}_{l \in {\cal L}} $ the eigenvectors of  ${\bf{\tilde G }}$ corresponding to eigenvalues $\{{\lambda _l}\}_{l \in {\cal L}}$. Let us define
\begin{align}
{ \boldsymbol{\bar \chi} _l} \buildrel \Delta \over = \mathop {\lim }\limits_{\omega  \to 0} {{\boldsymbol \chi} _l}, \quad\forall l \in {\cal L},
	\label{eq67}\end{align}
which is a basis of vector to span $L$ dimensional space. That is, each column of ${{\bf{Q}}_{A,1}}$ is a linear combination of vectors in a vector space $\{{ \boldsymbol{\bar \chi} _l}\}$ over a field $\{\beta _{l,i}\}$, i.e.,
\begin{align}
	{{\bf{q}}_{A,l}} = \sum\limits_{i \in {\cal L}} {{\beta _{l,i}}{{\boldsymbol{\bar \chi }}_i}},\quad \forall l \in {\cal L},
	\label{eq68}\end{align}
where the linear coefficients satisfy $\sum\nolimits_{i \in {\cal L}} {\beta _{l,i}^2}  = 1,\forall l$.

Then, we prove that the eigenvector $\{{ \boldsymbol{\bar \chi} _l}\}$ of matrix ${\bf{\tilde G}}$ corresponding to $\{{\lambda _l}\}$ converges to the $l$th column of ${{\bf{Q}}_{A,1}}$ by contradiction methods. Starting from $l=L$, we suppose ${{\bf{q}}_{A,L}} \ne {{\boldsymbol{\bar \chi }}_L}$, as $\omega  \to 0$. Then when ${\beta _{L,L}} \ne 1$ there exists an index ${i^{*}}({i^*} < L)$ such that ${\beta _{L,{i^*}}}\ne0$, and we have
\begin{eqnarray}
	\begin{aligned}[b]
		\mathop {\lim }\limits_{\omega  \to 0} {g _L}& = \mathop {\lim }\limits_{\omega  \to 0} {\bf{q}}_{A,L}^T{\bf{\tilde G }}{{\bf{q}}_{A,L}}\\&
		= \mathop {\lim }\limits_{\omega  \to 0} {\left( {\sum\limits_{i \in {\cal L}} {{\beta _{L,i}}{{\boldsymbol{\bar \chi }}_i}} } \right)^T}{\bf{\tilde G }}\left( {\sum\limits_{i \in {\cal L}} {{\beta _{L,i}}{{\boldsymbol{\bar \chi }}_i}} } \right)\\&
		= \mathop {\lim }\limits_{\omega  \to 0} \sum\limits_{i \in {\cal L}} {\beta _{L,i}^2} {\boldsymbol{\bar \chi }}_i^T{\bf{\tilde G }}{{\boldsymbol{\bar \chi }}_i}\\&
		= \mathop {\lim }\limits_{\omega  \to 0} \sum\limits_{i \in {\cal L}} {\beta _{L,i}^2} {\lambda _i} \ge \beta _{L,{i^*}}^2{\lambda _{{i^*}}}.
	\end{aligned}
	\label{eq69}\end{eqnarray}
Combining this and the result in~\eqref{eq66}, we have
\begin{align}
	\mathop {\lim }\limits_{\omega  \to 0} \frac{\ln ({g _L})}{\ln (\omega \cos \phi )} \!\le \!\mathop {\lim }\limits_{\omega  \to 0} \frac{\ln (\beta _{M,{i^*}}^2{\lambda _{i^*}})}{\ln (\omega\cos \phi )}\! =\! 2({i^*} \!-\! 1) < 2(L \!-\! 1).
	\label{eq70}\end{align}
Explicitly,~\eqref{eq70} contradicts with the result in~\eqref{eq59}, indicating the prior assumption is invalid. Now we can get ${{\bf{q}}_{A,L}} = {{\boldsymbol{\bar \chi }}_L}$.

Then, we recursively deduce to $l=L-1$. Based on ${{\boldsymbol{q}}_{A,L}} = {{\boldsymbol{{\bar \chi } }}_L}$ and ${{\boldsymbol{{\bar \chi } }}_L} \bot {{\boldsymbol{{\bar \chi } }}_{L - 1}}$, we claim ${\beta _{M-1,M}} = 0$. Similarly, assuming ${{\bf{q}}_{A,L-1}} \ne {{\boldsymbol{\bar \chi }}_{L-1}}$, we get
\begin{align}
	\mathop {\lim }\limits_{\omega  \to 0} \frac{\ln ({\lambda_{L-1}})}{\ln (\omega\cos \phi )} < 2(L - 2),
	\label{eq71}\end{align}
which is again inconsistent with the result in~\eqref{eq59}. By utilizing recursive argument, we can deduce
\begin{align}
	{{\bf{q}}_{A,l}} = {{\boldsymbol{\bar \chi }}_{l}}, \quad \forall l \in {\cal L}.
	\label{eq72}\end{align}

Above derivations prove that the right singular vectors of ${\bf{\tilde H}}$ converge to the $\{{\bf{q}}_{A,l}\}_{l \in {\cal L}}$ as ${\omega  \to 0}$. Likewise, the left singular vectors of ${\bf{\tilde H}}$ converge to the $\{{\bf{q}}_{G,l}\}_{l \in {\cal L}}$, which is the $l$th column of ${{\bf{Q}}_{G,1}}$. Then we have 
\begin{eqnarray}
	\begin{aligned}[b]
		\mathop {\lim }\limits_{\omega  \to 0} \frac{{\lambda _l}}{{(\omega \cos \phi )}^{2(l - 1)}}
		= \mathop {\lim }\limits_{\omega  \to 0} \frac{{{| {{\bf{q}}_{G,l}^T{\bf{\tilde H}}{{\bf{q}}_{A,l}}} |}^2}}{{(\omega \cos \phi )}^{2(l - 1)}}.
	\end{aligned}
	\label{eq73}\end{eqnarray}
Next, ${{\bf{D}}_G}$ and ${{\bf{D}}_A}$ is implemented QR decomposition as
\begin{align}
	\left\{{\begin{array}{*{20}{c}}
			{{{\bf{D}}_G} = {{\bf{Q}}_G}{{\bf{R}}_G}},\\
			{{{\bf{D}}_A} = {{\bf{Q}}_A}{{\bf{R}}_A}},
	\end{array}}\right.
	\label{eq74}\end{align}
where ${\bf{Q}}_A \in {\mathbb{R}^{L \times L}}$ and ${\bf{Q}}_G \in {\mathbb{R}^{N \times N}}$ are the unitary matrices, ${\bf{R}}_A \in {\mathbb{R}^{L \times \infty }}$ and ${\bf{R}}_G \in {\mathbb{R}^{N \times \infty}}$ are the upper triangular matrices. Let $r_{A,(l,n)}$ and $r_{G,(l,n)}$ indicate the $(l,n)$th entry in ${{\bf{R}}_A}$ and ${{\bf{R}}_G}$, respectively. Thus, by combining~\eqref{eq52}, ~\eqref{eq73} and~\eqref{eq74}, yields
\begin{eqnarray}
	\begin{aligned}[b]
		&\mathop {\lim }\limits_{\omega  \to 0} \frac{{\lambda _l}}{{(\omega \cos \phi )}^{2(l - 1)}}\\&
		= \mathop {\lim }\limits_{\omega  \to 0} \frac{{{| {{\bf{q}}_{G,l}^T{{\bf{D}}_G}{\bf{\Lambda}\bf{D}}_A^T{{\bf{q}}_{A,l}}} |}^2}}{{(\omega \cos \phi )}^{2(l - 1)}}.\\&
		= \mathop {\lim }\limits_{\omega  \to 0} {\left| {\frac{{{\bf{q}}_{G,l}^T{{\bf{Q}}_G}{{\bf{R}}_G}{\bf{\Lambda}}{\bf{R}}_A^T{\bf{Q}}_A^T{{\bf{q}}_{A,l}}}}{{(\omega \cos \phi )}^{l - 1}}} \right|^2}\\&
		= \mathop {\lim }\limits_{\omega  \to 0} {\left| {\frac{{\sum\limits_{n = l}^\infty  {{r_{G,(l,n)}}\frac{{(j\omega \cos \phi )}^{n - 1}}{{(n - 1)!}}{r_{A,(l,n)}}}}}{{(\omega \cos \phi )}^{l - 1}}} \right|^2}\\&
		= {\left[{\frac{{r_{G,l}}{r_{A,l}}}{(l - 1)!}} \right]^2},
	\end{aligned}
	\label{eq75}\end{eqnarray}
where ${r_{G,l}} = {r_{G,(l,l)}}$ and ${r_{A,l}} = {r_{A,(l,l)}}$ denote the diagonal entities, respectively, and the proof is complete.

\section{Proof of Theorem~\ref{theorem02}} \label{app2}
Based on the characteristic of QR decomposition, we know that ${r_{A,1}}$ is an all-one vector with the length $L$ as $k=1$. Following, we focus on the calculation of $k>1$. Use ${{\bf{D}}_{A,(k)}}$ to indicate the first $k$ columns of the Vandermonde matrix ${{\bf{D}}_{A}}$ in~\eqref{eq53}. Then combining~\eqref{eq74}, yields
\begin{eqnarray}
	\begin{aligned}[b]
		\det ({\bf{D}}_{A,(k)}^T{{\bf{D}}_{A,(k)}})& = \det ({\bf R}_{A,(k)}^T{\bf Q}_A^T{{\bf Q}_A}{{\bf R}_{A,(k)}})\\&
		= \det ({\bf R}_{A,(k)}^T{{\bf R}_{A,(k)}}) \\&
		\overset{{(c)}} = \prod\limits_{i = 1}^k {r_{A,i}^2},
	\end{aligned}
	\label{eq76}\end{eqnarray}
where the equality $(c)$ sets up by applying the property that determinant of triangular matrices equals the product of diagonal elements. Thus we get
\begin{align}
	r_{A,k}^2 = \frac{{\prod\limits_{i = 1}^k {r_{A,i}^2} }}{{\prod\limits_{i = 1}^{k - 1} {r_{A,i}^2} }} = \frac{{\det ({\bf{D}}_{A,(k)}^T{{\bf{D}}_{A,(k)}})}}{{\det ({\bf{D}}_{A,(k - 1)}^T{{\bf{D}}_{A,(k - 1)}})}}.
	\label{eq77}\end{align}
Following the Cauchy-Binet formula~\cite{Cauchy}, we obtain
\begin{align}
	\det ({\bf{D}}_{A,(K)}^T{{\bf{D}}_{A,(K)}}) = \sum\limits_{{\cal C}_K} {\det ({\bf{D}}_{A,[{\cal C}_K]}^T)\det ({{\bf{D}}_{A,[{\cal C}_K]}})},
	\label{eq78}\end{align}
where ${\bf{D}}_{A,[{\cal C}_K]}$ denotes a $K$-by-$K$ sub-matrix whose rows are the rows of ${\bf{D}}_{A,(K)}$ at indices belong to set ${{\cal C}_K}$. Since ${\bf{D}}_{A,[{\cal C}_K]}$ is a Vandermonde matrix, its determinant can be calculated as
\begin{align}
	\det ({{\bf{D}}_{A,[{{\cal C}_K}]}}) = \prod\limits_{\{i < j\}\in{{\cal C}_K}} {({\delta _j} - {\delta _i})}.
	\label{eq79}\end{align}
Therefore, by substituting~\eqref{eq79} into~\eqref{eq78} and~\eqref{eq77}, which completes the proof.

\section{Proof of Corollary~\ref{corol01}} \label{app3}
We introduce an associated fundamental (or cardinal) Lagrange interpolating polynomial related to Fekete points as
\begin{align}
	{{L}}_k(x) = \prod\limits_{i \ne k} {\frac{{x - {\mu _i}}}{{\mu _k} - {\mu _i}}}.
	\label{eq80}\end{align}
Any polynomial function whose degree less than $K$ can be written as
\begin{align}
	{\mathscr{L}}(x) = \sum\limits_{k \in {\cal K}} {{{L}}_k(x){\mathscr{L}}({\mu _k})}.
	\label{eq81}\end{align}
Thus, each entry in ${\bf{D}}_{A,(K)}$ can be expressed as
\begin{align}
	{\delta _l^{s}} = \sum\limits_{k \in {\cal K}} {{{{L}}_k}({\delta _l}){\mu _k^{s}}}.  
	\label{eq82}\end{align}
Next, we can write ${\bf{D}}_{A,(K)}$ in a more compact form as
\begin{align}
	{{\bf{D}}_{A,(K)}} = {\bf{\Psi}}{\bf{Z}},
	\label{eq83}\end{align}  
where ${\bf{\Psi}}=\{{{{L}}_k}({\delta _l})\}_{l \in {\cal L},k \in {\cal K}}$, and ${\bf{Z}}$ is a $K$-dimensional Vandermonde matrix determined by Fekete points ${\boldsymbol{\mu}} = ({\mu _1},{\mu _2},...,{\mu _K})$. In that way, \text{(P1.2)} is rewritten as a matrix form, i.e.,
\begin{align} 
	{\mathscr{D}}_{{\cal C}_K}({\boldsymbol{\delta}})=\sum\limits_{{\cal C}_K} {\prod\limits_{\{i < j\}\in {{\cal C}_K}} {({\delta _j} - {\delta _i})}^2}  = \det ({\bf{D}}_{A,(K)}^T{\bf{D}}_{A,(K)}).
	\label{eq84}\end{align}
Then, we substitute~\eqref{eq83} into~\eqref{eq84} , which leads to
\begin{eqnarray}
	\begin{aligned}[b]
		&\det ({\bf{D}}_{A,(K)}^T{{\bf{D}}_{A,(K)}})\\&
    	= \det \left( {{{({\bf{\Psi}}{\bf{Z}})}^T}({\bf{\Psi}}{\bf{Z}})} \right)\\&
		= \det \left( {{\bf{Z}}^T{\bf{\Psi}}^T}{\bf{\Psi}}{\bf{Z}} \right)\\&
		\overset{(d)} = \det ({\bf{Z}}^T)\det ({{\bf{\Psi}}^T}{\bf{\Psi}})\det ({\bf{Z}})\\&
    	= {\mathscr{D}}_{{\cal C}_K}^2({\boldsymbol{\mu }})\det ({{\bf{\Psi}}^T}{\bf{\Psi}})\\&
		\overset{(e)}\le {\mathscr{D}}_{{\cal C}_K}^2({\boldsymbol{\mu }}){\prod\limits_{k \in {\cal K}} {\left( {\sum\limits_{l \in {\cal L}} {L_k^2({\delta _l})} } \right)}}\\&
		\overset{(f)}\le {{\mathscr{D}}_{{\cal C}_K}^2}({\boldsymbol{\mu }}){\left( {\frac{1}{K}\sum\limits_{k \in {\cal K}} {\sum\limits_{l \in {\cal L}} {L_k^2({\delta _l})} } } \right)^K},
	\end{aligned}
	\label{eq85}\end{eqnarray}
where the equality $(d)$ holds since the matrices ${\bf{\Psi}}$ and ${\bf{Z}}$ are both invertible, the inequality $(e)$ is satisfied in terms of the Hadamard inequality~\cite{Hadamard}, and inequality $(f)$ can be established according to the arithmetic mean-geometric mean inequality~\cite{inequality}. 

Recall the stronger inequality of polynomial, i.e.,
\begin{align}
	\mathop {\max }\limits_{x \in [ - 1,1]} \sum\limits_{k \in {\cal K}} {{L_k}(x)}  = 1.
	\label{eq86}\end{align}
The optimality can be achieved when we set the variables as Fekete-Gauss-Lobatto points~\cite[LM. 2.1.]{Fekete}. Substituting~\eqref{eq86} into~\eqref{eq85}, yields the upper bound as
\begin{align}
	{\mathscr{D}}_{{\cal C}_K}({\boldsymbol{\delta}}) \le {{\mathscr{D}}_{{\cal C}_K}^2}({\boldsymbol{\mu }}){\left( {\frac{L}{K}} \right)^K}.
	\label{eq87}\end{align}
This completes the proof.

\bibliographystyle{IEEEtran}
\bibliography{IEEEabrv,Reference}

\vfill

\end{document}